\newcommand*{\QUANTUM}{}
\ifdefined\QUANTUM
\documentclass[a4paper,onecolumn,11pt,accepted=2024-11-19]{quantumarticle}
\pdfoutput=1
\else
\documentclass[11pt]{article}
\fi

\usepackage{color, xcolor, colortbl}
\usepackage{graphicx,epstopdf}
\usepackage{amsmath}
\usepackage{amssymb}
\usepackage{amsthm}
\usepackage{bm}
\usepackage{geometry}
\usepackage{appendix}
\usepackage{multirow}
\usepackage{braket}
\usepackage[english]{babel}
\usepackage{hyperref}
\usepackage[capitalize]{cleveref}
\usepackage{caption}
\usepackage{enumerate}
\usepackage{enumitem}
\crefrangeformat{enumi}{#3#1#4--#5#2#6}

\usepackage{algorithm}
\usepackage{algorithmic}
\usepackage{xspace}
\usepackage[qm]{qcircuit}
\usepackage{subcaption}
\usepackage{multirow}

\usepackage{newfloat}
\usepackage{etoolbox}
\usepackage{dsfont}

\usepackage{tikz-cd}
\usepackage{stmaryrd}

\newtheorem{theorem}{Theorem}
\newtheorem{definition}[theorem]{Definition}

\newtheorem{lemma}[theorem]{Lemma}
\newtheorem{corollary}[theorem]{Corollary}

\newtheorem{problem}[theorem]{Problem}

\newcommand{\bI}{\mathbb{I}}

\newcommand{\id}{\mathrm{id}}
\newcommand{\rd}{\mathrm{d}}

\newcommand{\bvec}[1]{\mathbf{#1}}
\DeclareMathAlphabet{\mymathbb}{U}{BOONDOX-ds}{m}{n}

\newcommand{\vg}{\bvec{g}}

\renewcommand{\Re}{\mathrm{Re}}
\renewcommand{\Im}{\mathrm{Im}}
\newcommand{\I}{\mathrm{i}}

\newcommand{\mc}[1]{\mathcal{#1}}

\newcommand{\wt}[1]{\widetilde{#1}}

\newcommand{\abs}[1]{\left\lvert#1\right\rvert}
\newcommand{\norm}[1]{\left\lVert#1\right\rVert}

\newcommand{\Or}{\mathcal{O}}

\newcommand{\NN}{\mathbb{N}}
\newcommand{\RR}{\mathbb{R}}
\newcommand{\CC}{\mathbb{C}}

\newcommand{\half}{\frac{1}{2}}

\DeclareMathOperator\arcsinh{arcsinh}
\DeclareMathOperator\arccosh{arccosh}

\usepackage[normalem]{ulem}

\begin{document}

\title{Infinite quantum signal processing}
\author{Yulong Dong} 
        \affiliation{Department of Mathematics, University of California, Berkeley,  CA 94720, USA.}
        \orcid{0000-0003-0577-2475}
\author{Lin Lin}
        \affiliation{Department of Mathematics, University of California, Berkeley,  CA 94720, USA.}
        \affiliation{Applied Mathematics and Computational Research Division, Lawrence Berkeley National Laboratory, Berkeley, CA 94720, USA}
        \affiliation{Challenge Institute for Quantum Computation, University of California, Berkeley,  CA 94720, USA}
        \orcid{0000-0001-6860-9566}
\author{Hongkang Ni} 
        \affiliation{
Institute for Computational and Mathematical Engineering (ICME), Stanford University, Stanford, CA 94305, USA.}
\author{Jiasu Wang} 
        \affiliation{Department of Mathematics, University of California, Berkeley, CA 94720, USA.}
        \orcid{0000-0002-1321-2649}

\begin{abstract}
Quantum signal processing (QSP) represents a real scalar polynomial of degree $d$ using a product of unitary matrices of size $2\times 2$, parameterized by $(d+1)$ real numbers called the phase factors. This innovative representation of polynomials has a wide range of applications in quantum computation. 
When the polynomial of interest is obtained by truncating an infinite polynomial series, a natural question is whether the phase factors have a well defined limit as the degree $d\to \infty$. 
While the phase factors are generally not unique, we find that there exists a consistent choice of parameterization so that the limit is well defined in the $\ell^1$ space. 
This generalization of QSP, called the infinite quantum signal processing, can be used to represent a large class of non-polynomial functions.
Our analysis  reveals a surprising connection between the regularity of the target function and the decay properties of the phase factors.
Our analysis also inspires a very simple and efficient algorithm to approximately compute the phase factors in the $\ell^1$ space. 
The algorithm uses only double precision arithmetic operations, and provably converges when the $\ell^1$ norm of the Chebyshev coefficients of the target function is upper bounded by a constant that is independent of $d$. This is also the first numerically stable algorithm for finding phase factors with provable performance guarantees in the limit $d\to \infty$.

\end{abstract}

\maketitle

\newpage
\tableofcontents
        
\section{Introduction}

\subsection{Background}\label{sec:background}
The study of the representation of polynomials has a long history, with rich applications in a diverse range of fields.
It is therefore exciting that a  new way of representing polynomials, called quantum signal processing\footnote{The term ``signal processing'' is due to an analogy to digital filter designs on classical computers.} (QSP)~\cite{LowChuang2017,GilyenSuLowEtAl2019}, has emerged recently in the context of quantum computation. 
The motivation for this development can be seen as follows.
For simplicity let $H\in\CC^{N\times N}$ be a Hermitian matrix with all eigenvalues in the interval $[-1,1]$, and let $f\in\RR[x]$ be a real scalar polynomial of degree $d$. We would like to efficiently encode  a matrix polynomial $f(H)$ using a unitary matrix that can be efficiently implemented on a quantum computer.
An inherent difficulty of this task is that a quantum algorithm is given by the product of a sequence of unitary matrices, but in general neither $H$ nor $f(H)$ is unitary. In an extreme scenario, let $H=x\in[-1,1]$ be a scalar, and we are interested in a representation of the polynomial $f(x)$ in terms of unitary matrices. 

Quantum signal processing proposes the following solution to the problem above: 
Let 
\begin{displaymath}
W(x) = e^{\I \arccos(x) X}=\left(\begin{array}{cc}{x} & {\I \sqrt{1-x^{2}}} \\ {\I \sqrt{1-x^{2}}} & {x}\end{array}\right), \text{ with } X=\begin{pmatrix}
0 & 1\\
1 & 0
\end{pmatrix},
\end{displaymath}
be a $2\times 2$ unitary matrix parameterized by $x\in[-1,1]$. 
Then the following expression 
\begin{equation}\label{eqn:qsp-unitary}
    U(x, \Psi) := e^{\I \psi_0 Z} \prod_{j=1}^{d} \left[ W(x) e^{\I \psi_j Z} \right], \text{ with } Z=\begin{pmatrix}
1 & 0 \\
0 & -1
\end{pmatrix},
\end{equation}
is a unitary matrix for any choice of phase factors $\Psi:=(\psi_0,\psi_1,\cdots,\psi_d)\in\RR^{d+1}$.
Here $X,Z$ are Pauli matrices. 
One can verify that the top-left entry of $U(x,\Psi)$ is a complex polynomial in $x$.
Moreover, for any \textit{target polynomial} $f\in\RR[x]$ satisfying (1) $\deg(f)=d$, (2) the  parity of $f$ is $d \bmod 2$, (3) $\norm{f}_{\infty}:=\max_{x\in[-1,1]} \abs{f(x)}\le 1$, we can find phase factors $\Psi\in\RR^{d+1}$ such that $f(x)$ is equal to the real (or imaginary) part of the top-left entry of $U(x,\Psi)$ for all $x\in[-1,1]$~\cite{GilyenSuLowEtAl2019}.
By setting $x=\frac{z+z^{-1}}{2}$ with $z=e^{\I \theta},\theta\in[0,2\pi)$, the representation in \cref{eqn:qsp-unitary} can be viewed as a  $2\times 2$  matrix Laurent polynomial, and we are interested in its values on the unit circle.

\cref{eqn:qsp-unitary} is an innovative way of encoding the information of a polynomial in terms of $2\times 2$ unitary matrices. 
It also leads to a very compact quantum algorithm for constructing a unitary matrix that encodes  the matrix polynomial $f(H)$, called \emph{quantum singular value transformation (QSVT)}~\cite{GilyenSuLowEtAl2019}. Assume that $H$ is accessed via its \emph{block encoding}
\[
W=\begin{pmatrix}
H & *\\
* & *
\end{pmatrix},
\]
where $W\in\CC^{MN\times MN}$ is a unitary matrix ($M\ge 2$), the matrix $H$ is its top-left $N\times N$ matrix subblock, and $*$ indicates matrix entries irrelevant to the current task. 
When given the phase factors $\Psi$ and the block encoding $W$, QSVT constructs a unitary $U\in \CC^{2MN\times 2MN}$ (by introducing one ancillary qubit) such that
\[
U=\begin{pmatrix}
f(H) & *\\
* & *
\end{pmatrix}.
\]
In other words, although $f(H)$ is not a unitary matrix, it can be block encoded by a unitary matrix of a larger size that can be efficiently implemented on quantum computers. This construction has found many applications, such as Hamiltonian simulation~\cite{LowChuang2017,GilyenSuLowEtAl2019},  linear system of equations~\cite{GilyenSuLowEtAl2019,LinTong2020,MartynRossiTanEtAl2021}, eigenvalue problems~\cite{LinTong2020a,DongLinTong2022}, Gibbs states preparation~\cite{GilyenSuLowEtAl2019}, Petz recovery channel~\cite{GilyenLloydMarvianEtAl2022}, benchmarking quantum systems~\cite{DongLin2021,DongWhaleyLin2021}, to name a few.
We refer interested readers to Refs.~\cite{GilyenSuLowEtAl2019,MartynRossiTanEtAl2021}. 

To implement QSVT, we need to efficiently calculate the phase factors $\Psi$ corresponding to a target polynomial of degree $d$. Many of the aforementioned applications are formulated as the evaluation of a matrix function $f(H)$, where $f:[-1,1]\to \RR$ is not a polynomial but a smooth function, which can be expressed as an infinite polynomial series (e.g., the Chebyshev polynomial series). To approximate $f(H)$, we need to first truncate the polynomial series to $f^{(d)}$ with a proper degree $d$ so that the difference between $f^{(d)}$ and $f$ is sufficiently small. 
Then for each $f^{(d)}$ we can find (at least) one set of phase factors $\Psi^{(d)}$. When $d$ is fixed, there has been significant progress in computing the phase factors in the past few years~\cite{GilyenSuLowEtAl2019,Haah2019,ChaoDingGilyenEtAl2020,DongMengWhaleyEtAl2021,Ying2022}.
The questions we would like to answer in this paper are as follows. 

\begin{enumerate}

\item As $d\to \infty$, can the phase factors $\{\Psi^{(d)}\}$ be chosen to have a well-defined limit $\Psi^{\star}$ in a properly chosen space? 

\item If $f$ is smooth, its Chebyshev coefficients decay rapidly. 
Does the tail of $\Psi^{\star}$ exhibit decay properties? If so, how is it related to the smoothness of $f$?  

\item Is there an efficient algorithm to approximately compute $\Psi^{\star}$?

\end{enumerate}
Our work is to first generalize QSP to represent smooth functions with a set of infinitely long phase factors, and we dub the resulting limit \emph{infinite quantum signal processing} (iQSP).

\subsection{Setup of the problem}

We follow the bra-ket notation widely used in quantum mechanics. Specifically, we define two ``kets'' as basis vectors of $\CC^2$, namely, 
\[
\ket{0} := \begin{pmatrix}
1\\0
\end{pmatrix}, \quad \ket{1} := 
\begin{pmatrix}
0\\ 1
\end{pmatrix}.
\]
 The ``bra'' can be viewed as row vectors induced from the corresponding ket by taking Hermitian conjugate. In the bra notation, $\bra{0} = (1,0), \bra{1} =(0,1)$. The inner product is written as $\braket{x|y} := (\ket{x}, \ket{y})$. 
Using this notation, the upper left element of $U(x,\Psi)$ in \cref{eqn:qsp-unitary} can be written as $\langle 0|U(x,\Psi)|0\rangle$. 
Direct calculation shows that the real part of $\langle 0|U(x,\Psi)|0\rangle$ can be recovered from the imaginary part by adding $\frac{\pi}{4}$ to both $\psi_0$ and $\psi_d$:
\begin{equation}\label{eq:re_im_equiv}
    \Re[\langle0|U(x,\Psi)|0\rangle]=\Im\left[\langle 0|e^{\I \frac{\pi}{4}Z}U(x,\Psi)e^{\I \frac{\pi}{4}Z}|0\rangle\right].
\end{equation}
For convenience, throughout this paper, we  focus on the imaginary part of $\langle 0|U(x,\Psi)|0\rangle$, which is denoted by $g(x,\Psi)$, i.e.,
\begin{equation}\label{eqn:qsp-polynomial}
    g(x,\Psi):=\Im[\langle 0|U(x,\Psi)|0\rangle].
\end{equation}

Due to the parity constraint, the number of degrees of freedom in a given target polynomial $f\in\RR[x]$ is only $\wt{d}:=\lceil \frac{d+1}{2} \rceil$. Therefore the phase factors $\Psi$ cannot be uniquely defined. 
To address this problem, Ref.~\cite{DongMengWhaleyEtAl2021} suggests that phase factors $\Psi:=(\psi_0,\psi_1,\ldots,\psi_d)$ can be restricted to be symmetric:
\begin{equation}
\psi_j=\psi_{d-j}, \quad \forall j = 0, 1, \cdots, d.
\label{eqn:symmetry_phase}
\end{equation}
Without loss of generality, we define the \emph{reduced phase factors} as follows,
\begin{equation}
    \Phi=(\phi_0,\phi_1,\ldots,\phi_{\wt{d}-1}):=\begin{cases}
    (\frac{1}{2}\psi_{\wt{d}-1}, \psi_{\wt{d}},\cdots,\psi_{d}), & d \text{ is even},\\
    (\psi_{\wt{d}}, \psi_{\wt{d}+1},\cdots,\psi_{d}), & d \text{ is odd}.\\
    \end{cases}
\end{equation}
The number of reduced phase factors is equal to $\wt{d}=\lceil \frac{d+1}{2} \rceil$, and matches the number of degrees of freedom in $f$.
With some abuse of notation, we identify $U(x,\Phi)$ with $U(x,\Psi)$ and
$g(x,\Phi)$ with $g(x,\Psi)$, and $\Phi$ is always referred to as the \textit{reduced phase factors} of a full set of phase factors $\Psi$. 
For a given target polynomial, the existence of the symmetric phase factors is proved in~\cite[Theorem 1]{WangDongLin2021}, but the choice is still not unique. However, near the trivial phase factors $\Phi=(0,\ldots,0)$, there exists a unique and consistent choice of symmetric phase factors called the maximal solution~\cite{WangDongLin2021}.

Let $\ell^1$ denote the set of all infinite dimensional vectors with finite 1-norm:
\begin{equation}
    \ell^1:=\{v=(v_0,v_1,\cdots): \norm{v}_1<\infty\}, \quad \norm{v}_1:=\sum_{k=0}^{\infty}\abs{ v_k},\quad v=(v_0,v_1,\cdots).
\end{equation}
The vector space $\ell^1$ is complete, i.e., every Cauchy sequence of points in $\ell^1$ has a limit that is also in $\ell^1$. 
Let $\RR^{\infty}$ be the set of all infinite dimensional vectors with only a finite number of nonzero elements.

\begin{definition}[Target function]
\label{def:target_function}
A target function $f:\RR\to \RR$ is an infinite Chebyshev polynomial series with a definite parity
\begin{equation}
f(x)=\begin{cases}
\sum_{j=0}^{\infty} c_j T_{2j}(x),& f \mbox{ is even},\\
\sum_{j=0}^{\infty} c_j T_{2j+1}(x), & f \mbox{ is odd},
\end{cases}
\label{eqn:f_expand}
\end{equation}
The coefficient vector $c=(c_0,c_1,\ldots)\in\ell^1$, and $f$ satisfies the norm constraint
\begin{equation}
\norm{f}_{\infty}=\max_{x\in[-1,1]} \abs{f(x)}\le 1.
\end{equation}
In other words, the set of even target functions is
\begin{equation}
    S_e = \left\{f:[-1,1]\to [-1,1] : f(x) = \sum_{j=0}^{\infty} c_j T_{2j}(x),\quad\sum_{j=0}^{\infty}\abs{c_j} < \infty\right\}, 
\end{equation}
and the set of odd target functions is
\begin{equation}
    S_o = \left\{f:[-1,1]\to [-1,1] : f(x) = \sum_{j=0}^{\infty} c_j T_{2j+1}(x),\quad\sum_{j=0}^{\infty}\abs{c_j} < \infty\right\}. 
\end{equation}

\end{definition}

If we truncate the Chebyshev coefficients to be $c^{(\wt{d})}=(c_0,c_1,\ldots,c_{\wt{d}-1},0,\ldots)\in\RR^{\infty}$, the corresponding Chebyshev polynomial $f^{(d)}$ is of degree $d$ (recall that $\wt{d}=\lceil \frac{d+1}{2} \rceil$ and hence $d$ is determined by $\wt{d}$ and the parity of the function).
Furthermore, $c\in\ell^{1}$ implies that
$\lim_{d\to \infty} \norm{f^{(d)}-f}_{\infty}=0$. 
Throughout the paper, $f^{(d)}$ will be referred to as a \textit{target polynomial} approximating the target function $f$.

In order to compare phase factors of different lengths, an important observation is that if we pad $\Phi=(\phi_0,\phi_1,\ldots,\phi_{\wt{d}})$ with an arbitrary number of $0$'s at the right end and obtain $\wt{\Phi}=(\phi_0,\phi_1,\ldots,\phi_{\wt{d}},0,\ldots,0)$, we have $g(x,\Phi)=g(x,\wt{\Phi})$ (see \cref{lm:pad_phi_0}).  Therefore $g(x,\cdot)$ is a well defined mapping in $\RR^{\infty}$, and we can identify $\Phi$ with $\wt{\Phi}$.
Let $\mc{F}$ be the linear mapping from a target polynomial to its Chebyshev-coefficient vector $c\in\RR^{\infty}$ as defined in \cref{eqn:f_expand}. This induces a mapping 
\begin{equation}\label{eq:F_dfn}
F:\RR^{\infty}\to \RR^{\infty}, \quad F(\Phi):= \mathcal{F}(g(x,\Phi)),
\end{equation}
which maps the reduced phase factors $\Phi\in\RR^{\infty}$ to the Chebyshev coefficients of $g(x,\Phi)$. 

Note that $\RR^{\infty}$ is dense in $\ell^1$, i.e., any point in $\ell^1$ is either a point in $\RR^{\infty}$ or a limit point of $\RR^{\infty}$. By exploiting nice properties that $F$ and its Jacobian matrix are Lipschitz continuous, we can define $\overline{F}: \ell^1\to \ell^1$ to be the extension of $F$, such that $\overline{F}(\Phi)$ agrees with $F(\Phi)$ for any $\Phi\in\RR^{\infty}$.
Then the problem of infinite quantum signal processing asks whether the inverse of the mapping $\overline{F}$ exists.
\begin{problem}[Infinite quantum signal processing]
For a target function in \cref{def:target_function} given by its Chebyshev coefficients $c\in\ell^1$, does there exist $\Phi^{\star}\in \ell^1$ such that $\overline{F}(\Phi^{\star})=c$? 
\label{prob:iqsp}
\end{problem}

\subsection{Main results}
\begin{theorem}[Invertibility of $\overline{F}$]
\label{thm:main_l1}
There exists a universal constant $r_c\approx 0.902$, so that  
        $\overline{F}$ has an inverse map $\overline{F}^{-1}: B(0,r_c)\subset \ell^1\to \ell^1$, where $B(a,r):=\{v\in \ell^1:\norm{v-a}_1<r\}$. 
\end{theorem}

\cref{thm:main_l1} provides a positive answer to \cref{prob:iqsp} as well as to the first question raised in \cref{sec:background}, when the 1-norm of the Chebyshev coefficients is upper bounded by a constant. 
Note that for a given target function $f$, we can always multiply it by a constant $C$, so that the $Cf$ satisfies the condition of \cref{thm:main_l1}. The main technical tools are a series of vector $1$-norm estimates of $F$, and  matrix $1$-norm estimates of the Jacobian matrix $DF$. These estimates do not explicitly depend on the length of phase factors, and can therefore be extended to $\overline{F}$. 
A more detailed statement of \cref{thm:main_l1} is \cref{thm:inverse_l1}, which will be presented in \cref{sec:extension}.

Since $\Phi^{\star}=(\phi_0,\phi_1,\ldots)\in\ell^1$, the tail of $\Phi^{\star}$ must exhibit decay properties, i.e.,  $\lim_{n\to \infty} \sum_{k>n}\abs{\phi_k}=0$. \cref{fig:numerical-pf-decay} in \cref{sec:numerics} shows that the tail decay of $\Phi^{\star}$ closely matches that of the Chebyshev coefficients $c$. 
The duality between the smoothness of a function and the decay of its Fourier / Chebyshev coefficients is well studied (see e.g. \cite[Chapter 7]{Trefethen2019Approximation}). However, it is surprising that the tail decay of the reduced phase factors can be directly related to the smoothness of the target function.
Such a behavior was first numerically observed in Ref. \cite{DongMengWhaleyEtAl2021}, in which an explanation of the phenomenon was also given in the perturbative regime. 
Using the tools developed in proving \cref{thm:main_l1}, we provide a refined and non-perturbative analysis of the tail decay in  \cref{thm:main_decay}.

 \begin{theorem}[Decay properties of reduced phase factors]\label{thm:main_decay}
        Given a target function $f$ with $\norm{c}_1< r_c$, and $\Phi^{\star} := \overline{F}^{-1}(c)=(\phi_0,\phi_1,\ldots)\in \ell^1$, then there exist constants $C, C'$ such that for any $n$,  
        \begin{equation}\label{eq:ineq-weak-decay}
                C'\sum_{k>n} \abs{c_k} \le \sum_{k>n}\abs{\phi_k}\le C\sum_{k>n} \abs{c_k}.
        \end{equation}
\end{theorem}
The proof is given in \cref{sec:decay} with an explicit characterization of constants $C,C'$.
Assume the target function $f$ is of $C^{\alpha}$ smoothness for some $\alpha > 0$, then the Chebyshev coefficients decay algebraically in the sense of $\sum_{k>n} \abs{c_k} = \Or\left(n^{-\alpha}\right)$. Then, it induces a decay of the corresponding reduced phase factors, namely, $\sum_{k>n}\abs{\phi_k} = \Or\left(n^{-\alpha }\right)$.  If $f$ is of $C^{\infty}$ or $C^{\omega}$ smoothness, then the tail of its Chebyshev-coefficient vector decays super-algebraically or exponentially respectively, and so does the tail of the reduced factors. We summarize this in the following corollary.
\begin{corollary}
    If $f$ is of $C^{\alpha}$ smoothness for some $\alpha > 0$, then the tail of the corresponding reduced phase factors decay algebraically, i.e., $\sum_{k>n}\abs{\phi_k} = \Or\left(n^{-\alpha }\right)$. 
    Furthermore, if $f$ is of $C^{\infty}$ or $C^{\omega}$ smoothness, then the tail of the corresponding reduced phase factors decays super-algebraically or exponentially respectively.
\end{corollary}
These results are also verified numerically in \cref{fig:numerical-pf-decay}.
These results provide a positive answer to the second question raised in \cref{sec:background}.

\cref{thm:main_l1} also has algorithmic implications. 
It has been empirically observed that a quasi-Newton optimization based algorithm is highly effective for finding the phase factors~\cite{DongMengWhaleyEtAl2021}. However, the theoretical justification of the optimization based algorithm has only been shown if the target function satisfies $\norm{f}_{\infty}\leq C_f \wt{d}^{-1}$, where $C_f\approx 0.028$ is a universal constant~\cite[Corollary 7]{WangDongLin2021}. Hence as $\wt{d}$ increases, existing theoretical results fail to predict the effectiveness of the algorithm, even if the target function is a fixed polynomial of finite degree (in this case, we pad the Chebyshev coefficients with zeros to increase $\wt{d}$).

Inspired by our analysis of the Jacobian map $DF$, we propose a very simple iterative algorithm to find phase factors for a given target polynomial $f$ (\cref{alg:iterative-qsp}). This algorithm can be viewed as finding the fixed point of the mapping $G(x):=x-\frac{1}{2}(F\left(x\right)-c)$ by means of a fixed point iteration $\Phi^{t+1}=G(\Phi^{t})$. 
This can also be viewed as an inexact Newton algorithm~\cite[Chapter 11]{NocedalWright1999}, as the inverse of the Jacobian matrix at $0\in\ell^{1}$ satisfies $[DF(0)]^{-1}=\frac12 \mathbb{I}$ (\cref{lm:DF0}).
\begin{algorithm}[htbp]
\caption{Fixed-point iteration algorithm for solving reduced phase factors}
\label{alg:iterative-qsp}
\begin{algorithmic}
\STATE{\textbf{Input:}  Chebyshev-coefficient vector $c$ of a target polynomial, and stopping criteria.}
\STATE{Initiate $\Phi^{0} = 0$, $t=0$;}
\WHILE{stopping criterion is not satisfied}
        \STATE{Update $\Phi^{t+1}=\Phi^{t}-\frac{1}{2}\left( F\left(\Phi^{t}\right)-c\right)$;}
        \STATE{Set $t=t+1$;
                }
\ENDWHILE
\STATE{\textbf{Output:} Reduced phase factors $\Phi^{t}$}.
\end{algorithmic}
\end{algorithm}

For a given target function, the Chebyshev-coefficient vector can be efficiently evaluated using the fast Fourier transform (FFT).
The fixed-point iteration algorithm (\cref{alg:iterative-qsp}) is the simplest algorithm thus far to evaluate phase factors. This algorithm provably converges when $\norm{c}_1$ is upper bounded by a constant. 

\begin{theorem}[Convergence of the fixed-point iteration algorithm]
There exists a universal constant $\wt{r}_c\approx 0.861$ so that when $\norm{c}_1 \le \wt{r}_c$, 

(i) \cref{alg:iterative-qsp} converges $Q$-linearly to $\Phi^{\star}=\overline{F}^{-1}(c)$, i.e., there exists a constant $C$ and the error satisfies
\begin{equation}
\norm{\Phi^{t}-\Phi^{\star}}_1 \le C \wt{\gamma}^{t}, \quad \wt{\gamma} \approx 0.8189 ,\quad t\ge 1.
\end{equation}
 
(ii) The overall time complexity is $\mathcal{O}(d^2\log\frac{1}{\epsilon})$, where $d$ is the degree of target polynomial and $\epsilon$ is the desired precision.  
 \label{thm:main_converge}
 \end{theorem}
A more accurate characterization about the region where \cref{alg:iterative-qsp} converges and the convergence rate is presented in \cref{sec:convergence}. 
In \cref{sec:numerics}, numerical experiments demonstrate that \cref{alg:iterative-qsp} is an efficient algorithm, and we observe that its convergence radius can be much larger than the theoretical prediction. These results provide a positive answer to the third question raised in \cref{sec:background}. We implement \cref{alg:iterative-qsp} with more examples as part of QSPPACK, an open-source package for finding phase factors~\footnote{The examples are available on the website \url{https://qsppack.gitbook.io/qsppack/} and the codes are open-sourced in \url{https://github.com/qsppack/QSPPACK}.}.

\subsection{Related works}

The original QSP paper~\cite{LowChuang2017} demonstrated the existence of phase factors without providing a constructive algorithm, and finding phase factors was considered to be a main bottleneck of the approach~\cite{ChildsMaslovNamEtAl2018}.
In the past few years, there has been significant progresses in computing the phase factors. 
Refs.~\cite{GilyenSuLowEtAl2019,Haah2019} developed the factorization based method. 
For a given target (real) polynomial $f^{(d)}$, one needs to find a complementary (real) polynomial satisfying the requirement of \cite[Corollary 5]{GilyenSuLowEtAl2019} (also see \cref{thm:qsp}). 
This step is based on finding roots of high degree polynomials to high precision, and this is not a numerically stable procedure. Specifically, the algorithm requires $\Or(d\log(d/\epsilon))$ bits of precision~\cite{Haah2019}. 
There have been two recent improvements of the factorization based method, based on the capitalization method~\cite{ChaoDingGilyenEtAl2020}, and the Prony method~\cite{Ying2022}, respectively. 
Although the two methods differ significantly, empirical results indicate that both methods are numerically stable, and are applicable to polynomials of large degrees. 
Furthermore, both methods take advantage of that the mapping from the Chebyshev coefficients $c$ to the full phase factors $\Psi$ \textit{is not necessarily} well-defined in $\ell^1$. For instance, a key step in \cite{ChaoDingGilyenEtAl2020} is to introduce a very small perturbation to the high order Chebyshev coefficients, which can nonetheless induce a large change in the phase factors $\Psi$. As a result, the question raised in \cref{prob:iqsp} cannot be well defined in such factorization based methods, and the tail of the phase factors $\Psi$ does not exhibit decay properties.

The optimization based method developed in Ref.~\cite{DongMengWhaleyEtAl2021} uses a different approach, and computes the symmetric phase factors without explicitly constructing the complementary polynomials. Empirical results show that the quasi-Newton optimization method in \cite{DongMengWhaleyEtAl2021} is numerically stable and can be applicable to polynomials of large degrees. 
Ref.~\cite{WangDongLin2021} analyzes the symmetric QSP, and proves that starting from a fixed initial guess of the reduced phase factors $\Phi=(0,\cdots,0)$, a simpler optimization method (the projected gradient method) converges linearly to a unique maximal solution, when the target polynomial satisfies $\norm{f^{(d)}}_{\infty}\le C d^{-1}$ for some constant $C$.
The fixed point iteration method in \cref{alg:iterative-qsp} is the simplest algorithm thus far for finding phase factors, and is the first provably numerically stable algorithm in the limit $d\to \infty$.

\section{Preliminaries on quantum signal processing}

The set  $[n]:=\{0,1,\cdots, n-1\}$ is referred to as the index set generated by a positive integer $n$. The row and column indices of a $n$-by-$n$ matrix run from $0$ to $n-1$, namely in the index set $[n]$. For a matrix $A\in\CC^{m\times n}$, the transpose, Hermitian conjugate and complex conjugate are denoted by $A^{\top}$, $A^{\dag}$, $A^*$, respectively. The same notations are also used for the operations on a vector.

For a matrix $M$ of infinite dimension, we equip it with the induced 1-norm, i.e.,
\begin{equation*}
    \norm{M}_1:=\max_{v\in \ell^1, \norm{v}_1=1}\norm{Mv}_1.
\end{equation*}

For any function $f$ over $[-1,1]$, we define its infinity norm as $\norm{f}_\infty := \max_{-1\le x \le 1} \abs{f(x)}$. The key to quantum signal processing (QSP) is a representation theorem for certain matrices  in $\text{SU}(2)$:
\begin{theorem}[\textbf{Quantum signal processing} {\cite[Theorem 4]{GilyenSuLowEtAl2019}}]\label{thm:qsp}
    For any $P, Q \in \mathbb{C}[x]$ and a positive integer $d$ such that
    \begin{enumerate}[label=(\arabic*)]
    \item $\deg(P) \leq d, \deg(Q) \leq d-1$,
    \item $P$ has parity $(d\mod2)$ and $Q$ has parity $(d-1 \mod 2)$,
    \item (Normalization condition) $|P(x)|^2 + (1-x^2) |Q(x)|^2 = 1, \forall x \in [-1, 1]$.
\end{enumerate}
Then, there exists a set of phase factors $\Psi := (\psi_0, \cdots, \psi_d) \in [-\pi, \pi)^{d+1}$ such that
\begin{equation}
\label{eq:qsp-gslw}
        U(x, \Psi) = e^{\I \psi_0 Z} \prod_{j=1}^{d} \left[ W(x) e^{\I \psi_j Z} \right] = \left( \begin{array}{cc}
        P(x) & \I Q(x) \sqrt{1 - x^2}\\
        \I Q^*(x) \sqrt{1 - x^2} & P^*(x)
        \end{array} \right)
\end{equation}
where 
\begin{displaymath}
W(x) = e^{\I \arccos(x) X}=\left(\begin{array}{cc}{x} & {\I \sqrt{1-x^{2}}} \\ {\I \sqrt{1-x^{2}}} & {x}\end{array}\right).
\end{displaymath}
\end{theorem}
Here, the complex conjugate of a complex polynomial is defined by taking complex conjugate on all of its coefficients. $X,Z$ are Pauli matrices. In most applications, we are only interested in using the real part of $P$. The following corollary is a slight variation of \cite[Corollary 5]{GilyenSuLowEtAl2019}, which states that the condition on the real part of $P$ can be easily satisfied.  Due to the relation between the real and imaginary components given in \cref{eq:re_im_equiv}, the conditions on the imaginary part of $P$ are the same.
\begin{corollary}[\textbf{Quantum signal processing with real target polynomials} {\cite[Corollary 5]{GilyenSuLowEtAl2019}}]\label{cor:complementary}
Let $f\in \RR[x]$ be a degree-d polynomial for some $d\geq 1$ such that 
\begin{itemize}
    \item $f(x)$ has parity $(d \mod 2)$,
    \item $\abs{f(x)}\le 1, \forall x\in[-1,1]$. 
\end{itemize}
Then there exists some $P,Q\in \CC[x]$ satisfying properties (1)-(3) of \cref{thm:qsp} such that $f(x)=\Im[P(x)]$.
\end{corollary}

Since we are interested in $P$, we may further restrict $Q\in\RR[x]$. In such a case, the phase factors can be restricted to be symmetric.
Let $D_d$ denote the domain of the symmetric phase factors:
\begin{equation}
    D_d=\begin{cases}
    [-\frac{\pi}{2},\frac{\pi}{2})^{\frac{d}{2}} \times [-\pi,\pi) \times [-\frac{\pi}{2},\frac{\pi}{2})^{\frac{d}{2}}, & d \mbox{ is even,}\\
    [-\frac{\pi}{2},\frac{\pi}{2})^{d+1}, &d \mbox{ is odd.}\\
    \end{cases}
\end{equation}

\begin{theorem}[\textbf{Quantum signal processing with symmetric phase factors} {\cite[Theorem 1]{WangDongLin2021}}]
\label{thm:sym_qsp}
Consider any $P\in \mathbb{C}[x]$ and $Q\in \mathbb{R}[x]$ satisfying the following conditions
\begin{enumerate}[label=(\arabic*)]
    \item $\deg(P)= d$ and $\deg(Q)= d-1$.
    \item $P$ has parity $(d \bmod 2)$ and $Q$ has parity $(d-1 \bmod 2)$.
    \item (Normalization condition) $\forall x\in[-1,1]: |P(x)|^2+(1-x^2)|Q(x)|^2=1$.
    \item \label{itm:4} If $d$ is odd, then the leading coefficient of $Q$ is positive.
\end{enumerate}
There exists a unique set of symmetric phase factors $\Psi:=(\psi_0,\psi_1,\cdots,\psi_1,\psi_0)\in D_d$ such that 
\begin{equation}\label{eq:UPQ}
U(x,\Psi)=\begin{pmatrix}
P(x) & \I Q(x)\sqrt{1-x^2}\\
\I Q(x) \sqrt{1-x^2} & P^* (x)
\end{pmatrix}.
\end{equation}
\end{theorem}
We want to emphasize that the set of symmetric phase factors is unique only if both $P(x)$ and $Q(x)$ are determined. If only $f = \Im [P]$ is given, then the set of symmetric phase factors may not be unique as mentioned above.
When we are only interested in $f(x)=\Im[P(x)]$ represented by symmetric phase factors, the conditions on $f$ are the same as those in \cref{cor:complementary}. 
This is proved constructively in \cite[Theorem 4]{WangDongLin2021}.

Throughout the paper, unless otherwise specified, we refer to $\Psi:=(\psi_0,\cdots,\psi_d)$ as the full set of phase factors and use $\Phi:=(\phi_0,\cdots,\phi_{\wt{d}-1})$ to denote the set of reduced phase factors after imposing symmetry constraint on $\Psi$, where $\wt{d}:=\lceil \frac{d+1}{2} \rceil$. 

In this paper, in order to characterize decay properties, we choose the second half of $\Psi$ to be the corresponding reduced phase factors. Specifically, when $d$ is odd, the set of reduced phase factors is
\begin{equation}
    \Phi=(\phi_0,\cdots,\phi_{\wt{d}-1}):=(\psi_{\wt{d}},\cdots,\psi_{d}).
\end{equation}
When $d$ is even, the set of reduced phase factors is
\begin{equation}
    \Phi=(\phi_0,\cdots,\phi_{\wt{d}-1}):=(\frac{1}{2}\psi_{\wt{d}-1}, \psi_{\wt{d}},\cdots,\psi_{d}).
\end{equation}

In this way, one has 
\begin{equation*}
        U(x,\Psi)=e^{i \phi_{\wt{d}-1} Z} W(x) e^{i \phi_{\wt{d}-2} Z}  \cdot \ldots \cdot W(x) e^{i \phi_{0} Z} W(x) e^{i \phi_{0} Z} W(x)\cdot \ldots\cdot e^{i \phi_{\wt{d}-2} Z} W(x) e^{i \phi_{\wt{d}-1} Z} 
\end{equation*}
for the odd case, and
\begin{equation*}
        U(x,\Psi)=e^{i \phi_{\wt{d}-1} Z} W(x) e^{i \phi_{\wt{d}-2} Z}  \cdot \ldots \cdot W(x) e^{2i \phi_{0} Z} W(x)\cdot \ldots\cdot e^{i \phi_{\wt{d}-2} Z} W(x) e^{i \phi_{\wt{d}-1} Z} 
\end{equation*}
for the even case. 

\begin{lemma}[Phase-factor padding]
\label{lm:pad_phi_0}
For any symmetric phase factors $\Psi$,
\begin{equation}
    \Im\left(\langle 0|U(x,\Psi)|0\rangle\right) = \Im \left(\langle 0| W(x) U(x,\Psi) W(x)|0\rangle\right).
\end{equation}
\end{lemma}
\begin{proof}
    According to the definition, $U(x,\Psi)$ takes the form $$\begin{pmatrix}
P(x) & \I Q(x)\sqrt{1-x^2}\\
\I Q(x) \sqrt{1-x^2} & P^* (x)
\end{pmatrix}.$$
Here, $P(x)\in \CC[x]$, and $Q(x)\in \RR[x]$ due to the symmetry of $\Psi$. Direct computation shows:
\begin{equation}
\begin{split}
    &W(x) U(x,\Psi) W(x)\\
    &= \begin{pmatrix}
x & \I \sqrt{1-x^2}\\
\I \sqrt{1-x^2} & x
\end{pmatrix}\begin{pmatrix}
P(x) & \I Q(x)\sqrt{1-x^2}\\
\I Q(x) \sqrt{1-x^2} & P^* (x)
\end{pmatrix} \begin{pmatrix}
x & \I \sqrt{1-x^2}\\
\I \sqrt{1-x^2} & x
\end{pmatrix}\\
&= \begin{pmatrix}
xP-(1-x^2)Q & \I (xQ+P^*)\sqrt{1-x^2}\\
\I (P+xQ) \sqrt{1-x^2} & xP^*-(1-x^2)Q
\end{pmatrix} \begin{pmatrix}
x & \I \sqrt{1-x^2}\\
\I \sqrt{1-x^2} & x
\end{pmatrix}.\\
\end{split}
\end{equation}
The upper-left entry of $W(x) U(x,\Psi) W(x)$ is
\begin{equation}
    x^2P-x(1-x^2)Q -(xQ+P^*)(1-x^2)= x^2P+(x^2-1)P^*-2x(1-x^2)Q.
\end{equation}
Hence
\begin{equation}
    \Im \left(\langle 0| W(x) U(x,\Psi) W(x)|0\rangle\right)= \Im \left(x^2P+(x^2-1)P^*\right) = x^2 P_\Im + (1-x^2)P_\Im = P_\Im.
\end{equation}
Note that $P_\Im = \Im\left(\langle 0|U(x,\Psi)|0\rangle\right)$, which completes the proof.
\end{proof}
Recall that we are interested in $g(x,\Psi):=\Im[\langle 0|U(x,\Psi)|0\rangle]$, and $g(x,\Phi)$ is identified with $g(x,\Psi)$. \cref{lm:pad_phi_0} implies that for reduced phase factors $\Phi$, $g(x,\Phi)$ remains the same if we pad $\Phi$ with an arbitrary number of  $0$'s at the right end. In this way, we are able to identify $\Phi$ with the infinite dimensional vector $(\phi_0,\cdots,\phi_{\wt{d}-1},0,\cdots)$ in $\RR^{\infty}$. Then for any $\Phi^{(1)},\Phi^{(2)}\in \RR^{\infty}$, the distance $\norm{\Phi^{(1)}-\Phi^{(2)}}_1$ is well defined. 

\begin{definition}
The effective length of $\Phi\in\RR^{\infty}$ is the largest index of the nonzero elements of $\Phi$.
If $\Phi=(\phi_0,\cdots,\phi_{\wt{d}-1},0,\cdots)$ and $\phi_{\wt{d}-1}\ne 0$, then its effective length is $\wt{d}$. 
\end{definition}

By viewing reduced phase factors $\Phi$ as an infinite dimensional vector in $\RR^{\infty}$, the problem of \emph{symmetric quantum signal processing} is to find reduced phase factors $\Phi\in \RR^{\infty}$ such that
\begin{equation}
    F(\Phi):=\mathcal{F}\left(g(x,\Phi)\right) = \mathcal{F}(f)
\end{equation}
holds for a target polynomial $f$.

\section{Infinite quantum signal processing}\label{sec:inf_sym_QSP}

We use $DF(\Phi)$ to denote the Jacobian matrix of $F(\Phi)$, which is a matrix of infinite dimension. Following the construction of the mapping $F$, for any $k\in\NN$, the $k$-th column of $DF(\Phi)$ is 
\begin{equation}\label{eqn:DF_column}
    \frac{\partial F(\Phi)}{\partial \phi_k}= \mathcal{F}\left(\frac{\partial g(x,\Phi)}{\partial \phi_k}\right).
\end{equation}
Similarly, for any $r,s\in\NN$, the second order derivative is
\begin{equation}\label{eqn:DF2_entry}
    \frac{\partial^2 F(\Phi)}{\partial \phi_r \partial \phi_s}= \mathcal{F}\left(\frac{\partial ^2 g(x,\Phi)}{\partial \phi_r \partial \phi_s}\right).
\end{equation}
As a remark, both $\frac{\partial F(\Phi)}{\partial \phi_k}$ and $\frac{\partial^2 F(\Phi)}{\partial \phi_r \partial \phi_s}$ are infinite dimensional vectors.

The main goal of this section is to prove \cref{thm:main_l1}.
We first present a useful estimate of the vector $1$-norm of $F$ and $\frac{\partial^2 F(\Phi)}{\partial \phi_r \partial \phi_s}$ in \cref{sec:1-norm-bounds-on-F}. 
This allows us to estimate the matrix $1$-norm of the Jacobian $DF$ and prove the invertibility of $DF$ in \cref{sec:lipschitz}. 
Based on these technical preparations, we prove the invertibility of the mapping $F$ in $\RR^{\infty}$ in \cref{sec:invert_F_Rinf}. 
As a final step, since $\RR^{\infty}$ is dense in $\ell^1$ and all derived estimates are independent of the effective length of $\Phi$, we extend the result to the invertibility of $\overline{F}$ in \cref{sec:extension}.  
The analysis leverages some facts about the Banach space, which are summarized in \cref{sec:banach} for completeness.

Without loss of generality, we consider the case that the target function is even in this section. The analytical results can be similarly generalized to the odd case.

\subsection{Estimating the vector 1-norm of $F$ and its second-order derivatives}\label{sec:1-norm-bounds-on-F}

We first summarize the main goal of this subsection as the following lemmas. To prove them, we consider a more general case where phase factors are not necessarily symmetric. Consequentially, we prove stronger results in \cref{lm:upper_bound_1_norm,lm:hess_1_norm}. \cref{lm:sym_upper_bound_1_norm} and \cref{lm:sym_hess_1_norm} are consequences of \cref{lm:upper_bound_1_norm} and \cref{lm:hess_1_norm} respectively by restricting to the symmetric phase factors. As a remark, the upper bounds are independent of the effective length of the reduced phase factors, which will enable the generalization to $\ell^1$.
    \begin{lemma}\label{lm:sym_upper_bound_1_norm}
            For any $\Phi\in \RR^{\infty}$, it holds that
            \begin{equation}
                    \norm{F(\Phi)}_1\le \sinh\left(2\norm{\Phi}_1\right).
            \end{equation}
    \end{lemma}
    \begin{lemma}\label{lm:sym_hess_1_norm}
            For any $\Phi\in\RR^{\infty}$, and $r,s\in \NN$, it holds that 
            \begin{equation}
                    \norm{\frac{\partial^2 F(\Phi)}{\partial \phi_r \partial \phi_s}}_1\leq 4\sinh\left(2\norm{\Phi}_1\right).
            \end{equation}
    \end{lemma}

To prove the previous lemmas, we start from a general setup where the phase factors are not necessarily symmetric.
\begin{lemma}\label{lm:upper_bound_1_norm}
For any full set of phase factors $\Psi:=\left(\psi_0,\psi_1,\cdots,\psi_d\right)$, it holds that 
\begin{equation}
    \norm{\mathcal{F}\left(g(x,\Psi)\right)}_1\leq \sinh\left(\norm{\Psi}_1\right).
\end{equation}
\end{lemma} 

\begin{proof}
Let $\Psi=\left(\psi_0,\psi_1,\cdots,\psi_d\right)$ be a full set of phase factors. The corresponding QSP matrix can be expanded as
\begin{equation}
    \begin{split}
        U(x,\Psi):&=e^{\I \psi_0 Z}\prod_{k=1}^d \left(W(x) e^{\I \psi_k Z}\right)\\
        &=\left(\cos( \psi_0)\bI+\I \sin(\psi_0) Z\right)\prod_{k=1}^d \left(W(x) \left(\cos( \psi_k)\bI+\I \sin(\psi_k) Z\right)\right)\\
        &=\left(\prod_{k=0}^d \cos(\psi_k)\right) \left(\bI+\I t_0 Z\right) \prod_{j=1}^d\left( W(x)\left(\bI+\I t_j Z\right)\right),
    \end{split}
\end{equation}
where $t_j:=\tan{\psi_j}$. Notice that $W(x)Z=Z W(x)^{-1}$ and
\begin{equation}\label{eqn:W_power}
    W(x)^k=\begin{pmatrix}
T_k(x) &\I \sqrt{1-x^2} U_{k-1}(x)\\
\I \sqrt{1-x^2} U_{k-1}(x) & T_k(x)
\end{pmatrix},
\end{equation}
where $T_k(x)$ and $U_k(x)$ are Chebyshev polynomials of the first and second kind respectively. Then,
\begin{equation}
    \begin{split}
        U(x,\Psi)&=\left(\prod_{k=0}^d \cos(\psi_k)\right) \left(\bI+\I t_0 Z\right)\prod_{j=1}^d \left(W(x) \left(\bI+\I t_j Z\right)\right)\\
        &=\left(\prod_{k=0}^d \cos(\psi_k)\right)\Big(W^d(x)+ \I  \sum_{j_1=0}^d t_{j_1} Z W^{d-2j_1}(x)-\sum_{j_1<j_2} t_{j_1} t_{j_2} W^{d-2(j_2-j_1)}(x) \\
        & \quad-\I \sum_{j_1<j_2<j_3} t_{j_1} t_{j_2} t_{j_3} Z W^{d-2(j_3-j_2+j_1)}(x)+...\Big).
    \end{split}
\end{equation}
Note that each term is a matrix whose upper left element is of the following form 
\begin{equation*}
    \I^l\cdot t_{j_1} t_{j_2}\cdots t_{j_l} T_k(x),
\end{equation*}
where $j_1<j_2<\cdots<j_l $ for some $l $, and $k = d-2\sum_{i=1}^l(-1)^{l-i}j_i$. 

When considering the Chebyshev coefficients of the imaginary part of the upper left element of $U(x,\Psi)$, only those terms with odd number of $t_j$'s matter. Then we have the estimate
\begin{equation}\label{eq:F_1_norm}
    \begin{split}
        \norm{\mathcal{F}\left(g(x,\Psi)\right)}_1&\leq \prod_{k=0}^d \abs{\cos(\psi_k)}\sum_{l \text{ is odd}}\sum_{j_1<j_2<\cdots<j_l} \prod_{i=1}^{l}\abs{t_{j_i}}\\
        &=\sum_{l \text{ is odd}}\sum_{j_1<j_2<\cdots<j_l} \prod_{i=1}^{l}\abs{\sin(\psi_{j_i})}\prod_{k\ne j_1,\cdots,j_l} \abs{\cos(\psi_k)}\\
        &\leq \sum_{l \text{ is odd}}\sum_{j_1<j_2<\cdots<j_l} \prod_{i=1}^{l}\abs{\sin(\psi_{j_i})}\\
        &\leq \sinh\left(\sum_{k=0}^d \abs{\sin(\psi_k)}\right)\\
        &\leq \sinh\left(\norm{\Psi}_1\right). 
    \end{split}
\end{equation}
We remark that the last inequality holds because $\sinh(x)$ is monotonic increasing, and the penultimate equality is due to the following observation:
\begin{equation}
    \begin{split}
        & \sinh\left(\sum_{k=0}^d \abs{\sin(\psi_k)}\right)=\sum_{l \text{ is odd}} \frac{1}{l!}\left(\sum_{k=0}^d \abs{\sin(\psi_k)}\right)^l =\sum_{l \text{ is odd}} \frac{1}{l!}\sum_{j_1,j_2,\cdots, j_l} \prod_{i=1}^{l}\abs{\sin(\psi_{j_i})}\\
        &\geq \sum_{l \text{ is odd}} \frac{1}{l!}\sum_{\substack{j_1,\cdots, j_l \\\text{ are all different}}} \prod_{i=1}^{l}\abs{\sin(\psi_{j_i})} =  \sum_{l \text{ is odd}} \sum_{j_1<j_2<\cdots<j_l} \prod_{i=1}^{l}\abs{\sin(\psi_{j_i})}.
    \end{split}
\end{equation}
The proof is completed.
\end{proof}

\begin{corollary}\label{lm:hess_1_norm}
For any full set of phase factors $\Psi:=\left(\psi_0,\psi_1,\cdots,\psi_d\right)$ and any $r,s \in [d+1]$, it holds that 
\begin{equation}
    \norm{\mathcal{F}\left(\frac{\partial^2 g(x,\Psi)}{\partial \psi_r\partial \psi_s}\right)}_1\leq \sinh(\norm{\Psi}_1).
\end{equation}
\end{corollary}
\begin{proof}
If $r=s$, then $\partial_r^2 U(x,\Psi)=-U(x,\Psi)$, and the desired result can be obtained by directly applying \cref{lm:upper_bound_1_norm}. Then we consider the case $r\ne s$. Note that
\begin{equation}
    \partial_r\partial_s U(x,\Psi)=U\left(x,\Psi+\frac{\pi}{2}e_r+\frac{\pi}{2}e_s\right),
\end{equation}
where $e_r=(0,\ldots,0,1,0,\ldots,0)$ denotes the $r$-th standard unit vector. Then
\begin{equation}
\begin{split}
    \frac{\partial^2 g(x,\Psi)}{\partial \psi_r\partial \psi_s}&=\frac{\partial^2 \Im[\langle 0| U(x,\Psi)|0\rangle]}{\partial \psi_r\partial \psi_s}=\Im[\langle 0| \partial_r\partial_s U(x,\Psi)|0\rangle]\\
    &=\Im[\langle 0| U(x,\Psi+\frac{\pi}{2}e_r+\frac{\pi}{2}e_s)|0\rangle]=g(x,\Psi+\frac{\pi}{2}e_r+\frac{\pi}{2}e_s).
\end{split}
\end{equation}
To simplify the notation, let $\tilde{\Psi} = \Psi+\frac{\pi}{2}e_r+\frac{\pi}{2}e_s$, and $\tilde{\psi}_k$ be the components of $\tilde{\Psi}$. Similar to the proof of \cref{lm:upper_bound_1_norm}, we have
\begin{equation}
    \begin{split}
        &\norm{\frac{\partial^2 g(x,\Psi)}{\partial \psi_r\partial \psi_s}}_1= \norm{g(x,\tilde{\Psi})}_1\\
        &\leq\sum_{l \text{ is odd}}\sum_{j_1<j_2<\cdots<j_l} \prod_{i=1}^{l}\abs{\sin(\tilde{\psi}_{j_i})}\prod_{k\ne j_1,\cdots,j_l} \abs{\cos(\tilde{\psi}_k)}\\
        &= \frac{1}{2}\prod_{k}\left(\abs{\cos(\tilde{\psi}_k)}+\abs{\sin(\tilde{\psi}_k)}\right)-\frac{1}{2}\prod_{k}\left(\abs{\cos(\tilde{\psi}_k)}-\abs{\sin(\tilde{\psi}_k)}\right)\\
        &\stackrel{(*)}{=}\sum_{l \text{ is odd}}\sum_{j_1<j_2<\cdots<j_l} \prod_{i=1}^{l}\abs{\sin(\psi_{j_i})}\prod_{k\ne j_1,\cdots,j_l} \abs{\cos(\psi_k)}\\
        &\leq  \sum_{l \text{ is odd}}\sum_{j_1<j_2<\cdots<j_l} \prod_{i=1}^{l}\abs{\sin(\psi_{j_i})}\\
        &\leq\sinh\left(\norm{\Psi}_1\right).
    \end{split}
\end{equation}
To see that equality $(*)$ holds, we note that $|\sin(\tilde{\psi}_k)|= |\cos({\psi}_k)|$ and $|\cos(\tilde{\psi}_k)| = |\sin({\psi}_k)|$ for $k=r\text{ or }s$ which interchanges two pairs of sine and cosine. Because this interchange operation leaves the production invariant, we can directly replace $\tilde \Psi$ by $\Psi$ in the expression. The proof is completed.
\end{proof}
        
        \cref{lm:sym_upper_bound_1_norm} is a direct application of \cref{lm:upper_bound_1_norm}. Now we use \cref{lm:hess_1_norm} to prove \cref{lm:sym_hess_1_norm}.
        \begin{proof}[Proof of \cref{lm:sym_hess_1_norm}]
            Choose $n\geq \max(r,s)$ such that all elements of $\Phi$ with index  $> n$ are zero. Then we may view $\Phi$ as a vector of length $n+1$, i.e., $(\phi_0,\cdots,\phi_n)$. Since we only consider the even case, we let $\Psi:=(\psi_0,\psi_1,\cdots,\psi_{2n})$ be the corresponding full set of phase factors and then
            \begin{equation}
                \psi_k=\begin{cases}
                \phi_{n-k} & k<n,\\
                2\phi_0 &k=n,\\
                \phi_{k-n} & k>n.
                \end{cases}
            \end{equation}
            For first-order derivative, when $k>0$,
            \begin{equation}
                \frac{\partial g(x,\Phi)}{\partial \phi_k}=\frac{\partial g(x,\Psi)}{\partial \psi_{n+k}}+\frac{\partial g(x,\Psi)}{\partial \psi_{n-k}},
            \end{equation} 
            and when $k=0$,
            \begin{equation}
                \frac{\partial g(x,\Phi)}{\partial \phi_k}=2\frac{\partial g(x,\Psi)}{\partial \psi_n}=\frac{\partial g(x,\Psi)}{\partial \psi_{n+k}}+\frac{\partial g(x,\Psi)}{\partial \psi_{n-k}}.
            \end{equation}
            Similarly, the second-order derivative is
            \begin{equation}
                \frac{\partial^2 g(x,\Phi)}{\partial \phi_r \partial \phi_s}=\frac{\partial^2 g(x,\Psi)}{\partial \psi_{n+r} \partial \psi_{n+s}}+\frac{\partial^2 g(x,\Psi)}{\partial \psi_{n-r} \partial \psi_{n+s}}+\frac{\partial^2 g(x,\Psi)}{\partial \psi_{n+r} \partial \psi_{n-s}}+\frac{\partial^2 g(x,\Psi)}{\partial \psi_{n-r} \partial \psi_{n-s}}.
            \end{equation}
            Invoking the triangle inequality for 1-norm and applying \cref{lm:hess_1_norm}, the results follow
            \begin{equation}
            \begin{split}
                \norm{\frac{\partial^2 F(\Phi)}{\partial \phi_r \partial \phi_s}}_1&=\norm{\mathcal{F}\left(\frac{\partial^2 g(x,\Phi)}{\partial \phi_r \partial \phi_s}\right)}_1\\
                &\leq\norm{\mathcal{F}\left(\frac{\partial^2 g(x,\Psi)}{\partial \psi_{n+r} \partial \psi_{n+s}}\right)}_1+\norm{\mathcal{F}\left(\frac{\partial^2 g(x,\Psi)}{\partial \psi_{n-r} \partial \psi_{n+s}}\right)}_1\\
                &\quad+\norm{\mathcal{F}\left(\frac{\partial^2 g(x,\Psi)}{\partial \psi_{n+r} \partial \psi_{n-s}}\right)}_1+\norm{\mathcal{F}\left(\frac{\partial^2 g(x,\Psi)}{\partial \psi_{n-r} \partial \psi_{n-s}}\right)}_1\\
                &\leq 4\sinh\left(2\norm{\Phi}_1\right).
            \end{split}
            \end{equation}
        \end{proof}

\subsection{Matrix $1$-norm estimates of $DF$}\label{sec:lipschitz}

The following lemma characterizes the Lipschitz continuity of $DF$.
\begin{lemma}[Lipschitz continuity of $DF$]\label{lm:DF_lipschitz}
        For any $\delta>0$ and any $\Phi^{(j)}\in \RR^{\infty}$ with $\norm{\Phi^{(j)}}_1\le \delta$, $j=1,2$, it holds that
        \begin{equation}
                \norm{DF(\Phi^{(1)}) - DF(\Phi^{(2)})}_1\le C_2(\delta) \norm{\Phi^{(1)}-\Phi^{(2)}}_1,
        \end{equation}
        where $C_2(\delta) = 4\sinh(2\delta)$. 
\end{lemma}     

\begin{proof}
Using the definition of the 1-norm of infinite dimensional matrix, one has
    \begin{equation}
        \begin{split}
            \norm{DF(\Phi^{(1)}) - DF(\Phi^{(2)})}_1&=\max_{\norm{v}_1=1} \norm{\left(DF(\Phi^{(1)}) - DF(\Phi^{(2)})\right)v}_1\\
            &=\max_{\norm{v}_1=1}\norm{\sum_{k=0}^{\infty}\left(\frac{\partial F(\Phi^{(1)})}{\partial \phi_k}-\frac{\partial F(\Phi^{(2)})}{\partial \phi_k}\right)v_k }_1\\
            &\le \max_{\norm{v}_1=1} \sum_{k=0}^{\infty} \norm{\frac{\partial F(\Phi^{(1)})}{\partial \phi_k}-\frac{\partial F(\Phi^{(2)})}{\partial \phi_k}}_1 \abs{v_k}\\
            &= \max_{k} \norm{\frac{\partial F(\Phi^{(1)})}{\partial \phi_k}-\frac{\partial F(\Phi^{(2)})}{\partial \phi_k}}_1.
        \end{split}
    \end{equation}
    For fixed $k$, $\Phi^{(1)}$ and $\Phi^{(2)}$, applying mean value inequality to the function
    \begin{equation*}
        y(t):=\frac{\partial F}{\partial \phi_k}\left(\Phi^{(1)}+t\left(\Phi^{(2)}-\Phi^{(1)}\right)\right),
    \end{equation*}
    we have
    \begin{equation}
        \begin{split}
                \norm{\frac{\partial F(\Phi^{(1)})}{\partial \phi_k}-\frac{\partial F(\Phi^{(2)})}{\partial \phi_k}}_1
                &\le\max_{\substack{\Phi' = (1-t)\Phi^{(1)}+t\Phi^{(2)}\\0\le t\le 1}}\norm{\nabla\frac{\partial F(\Phi')}{\partial \phi_k}\cdot\left(\Phi^{(1)}-\Phi^{(2)}\right)}_1\\
                &=\max_{\substack{\Phi' = (1-t)\Phi^{(1)}+t\Phi^{(2)}\\0\le t\le 1}}\norm{\sum_{l=0}^{n} \frac{\partial^2 F(\Phi')}{\partial\phi_l\partial \phi_k} \left(\Phi^{(1)}_l-\Phi^{(2)}_l\right)}_1\\
                &\le \max_{\substack{\Phi' = (1-t)\Phi^{(1)}+t\Phi^{(2)}\\0\le t\le 1}} \sum_{l=0}^n \norm{\frac{\partial^2 F(\Phi')}{\partial\phi_l\partial \phi_k}}_1 \abs{\Phi^{(1)}_l-\Phi^{(2)}_l}\\
                &\le C_2(\delta)\sum_{l=0}^{n} \abs{\Phi^{(1)}_l-\Phi^{(2)}_l}\\
                &=C_2(\delta)\norm{\Phi^{(1)}-\Phi^{(2)}}_1
        \end{split}
    \end{equation}
    where $n$ is the effective length of $\Phi^{(1)}-\Phi^{(2)}$. The last inequality follows \cref{lm:sym_hess_1_norm}, and notice that $\norm{\Phi'}_1 = \norm{(1-t)\Phi^{(1)}+t\Phi^{(2)}}_1$ is still bounded by $\delta$. Since $k$ can be arbitrary, the proof is completed.
\end{proof}     

\begin{lemma}\label{lm:DF0}
    $DF(0)=2\bI$, where $0\in\RR^{\infty}$ is the vector with all elements equal to zero, and $\bI$ is the identity matrix of infinite dimension.
\end{lemma}

\begin{proof}
    It is equivalent to show that $\frac{\partial F(0)}{\partial \phi_k}= 2e_k$, where $e_k\in\RR^{\infty}$ is the vector with all components equal to $0$ except for the $k$-th component which is equal to $1$. Recall \cref{eqn:DF_column} and notice that $\mathcal{F}(T_{2k}(x))= e_k$, we can prove that $\frac{\partial g(x,0)}{\partial \phi_k}= 2T_{2k}(x)$ instead. Invoking \cref{lm:pad_phi_0}, we only need to show that $\frac{\partial g(x,\Phi)}{\partial \phi_k}= 2T_{2k}(x)$, where $\Phi=(0,\cdots,0)\in \RR^{k+1}$. We know that $U(x,\Phi)= W^{2k}(x)$ as well as \cref{eqn:W_power}. Then direct computation gives that
    \begin{equation}
    \begin{split}
        \frac{\partial g(x,\Phi)}{\partial \phi_k} &= \Im \left(\braket{0|\frac{\partial U(x,\Phi)}{\partial \phi_k}|0}\right)=\Im \left(\braket{0|\I Z U(x,\Phi) + U(x,\Phi)\I Z|0}\right)\\
        &= 2\Re \left(\braket{0|U(x,\Phi)|0}\right)= 2\Re \left(\braket{0| W^{2k}(x)|0}\right)= 2T_{2k}(x).
    \end{split}
    \end{equation}
\end{proof}
By choosing $\Phi^{(2)}$ to be $0$ in \cref{lm:DF_lipschitz}, we obtain a rough estimate about $DF(\Phi)$,
\begin{equation*}
    \norm{DF(\Phi)-2\bI}_1\le C_2(\norm{\Phi}_1) \norm{\Phi}_1.
\end{equation*}
However, this estimate can be refined, which is given in the following lemma.
    \begin{lemma}\label{lm:DF-2I_estimate}
                Define
                \begin{equation}
                        h(x) := \int_0^x C_2(\delta)\rd \delta = 2\cosh(2x)-2.
                \end{equation} 
                For any $\Phi\in \RR^{\infty}$, 
                \begin{equation}
                        \norm{DF(\Phi)- 2\bI}_1\le h(\norm{\Phi}_1).
                \end{equation}
        \end{lemma}

\begin{proof}
    Note that $DF(0) = 2\bI$. For a fixed $\Phi$, for an arbitrary partition of $[0,1]$
    \begin{equation*}
        0=t_0<t_1<\cdots<t_m=1,
    \end{equation*}
    the following can be obtained by invoking the triangle inequality of 1-norm and applying \cref{lm:DF_lipschitz}
    \begin{equation}
    \begin{split}
        \norm{DF(\Phi)-DF(0)}_1&=\norm{\sum_{k\in[m]}\left(DF(t_{k+1}\Phi)-DF(t_k\Phi)\right)}_1\\
        &\leq \sum_{k\in[m]}\norm{DF(t_{k+1}\Phi)-DF(t_k\Phi)}_1\\
        &\leq \sum_{k\in[m]}C_2(\norm{t_{k+1}\Phi}_1)\norm{t_{k+1}\Phi-t_k\Phi}_1\\
        &=\sum_{k\in[m]}C_2(t_{k+1}\norm{\Phi}_1)\norm{\Phi}_1\abs{t_{k+1}-t_k}.
    \end{split}
    \end{equation}
    Note that $\sum_{k\in[m]}C_2(\norm{t_{k+1}\Phi}_1)\norm{\Phi}_1\abs{t_{k+1}-t_k}$ is a Riemann sum integrating $C_2(t\norm{\Phi}_1)\norm{\Phi}_1$ over $t \in [0,1]$. Since the partition is arbitrary, one gets
    \begin{equation}
        \norm{DF(\Phi)-DF(0)}_1\leq \int_{0}^1 C_2(t\norm{\Phi}_1)\norm{\Phi}_1 \rd t=\int_0^{\norm{\Phi}_1} C_2(\delta)\rd \delta=h(\norm{\Phi}_1).
    \end{equation}
\end{proof}
As an immediate consequence of  \cref{lm:DF-2I_estimate}, for any $\Phi\in \RR^{\infty}$ with bounded norm $\norm{\Phi}_1\le \delta$, it holds that
\begin{equation}
\label{eqn:def_C1}
  \norm{DF(\Phi)}_1\le 2\cosh(2x)-2 + 2=2\cosh(2\delta):=C_1(\delta).
\end{equation}
We will use this upper bound on $\norm{DF(\Phi)}_1$ to prove the Lipschitz continuity of $F$.
        \begin{corollary}[Lipschitz continuity of $F$]\label{re:equiv_dist_upper}
        For any $\Phi^{(j)}\in \RR^{\infty}$ with bounded norm $\norm{\Phi^{(j)}}_1\le \delta$ where $j=1,2$ it holds that
                \begin{equation}
                        \norm{F(\Phi^{(1)})-F(\Phi^{(2)})}_1\le C_1(\delta)\norm{\Phi^{(1)}-\Phi^{(2)}}_1.
                \end{equation}
        \end{corollary}
\begin{proof}
Applying the mean value inequality, there exists $\Phi'$ which is some convex combination of $\Phi^{(1)}$ and $\Phi^{(2)}$ so that
        \begin{equation}
        \norm{F(\Phi^{(1)})-F(\Phi^{(2)})}_1\leq \norm{DF(\Phi')}_1 \norm{\Phi^{(1)}-\Phi^{(2)}}_1.
        \end{equation}
        Invoking \cref{eqn:def_C1} and recalling the condition $\norm{\Phi'}_1\leq \delta$, one has $\norm{DF(\Phi')}_1\leq C_1(\delta)$, which completes the proof.
\end{proof} 

\subsection{Invertibility of $F$  in $\RR^{\infty}$}
\label{sec:invert_F_Rinf}
According to the inverse mapping theorem and \cref{lm:DF0}, we know that
$F$ is invertible near 0. We now prove a stronger result about the invertibility of $F$ in a neighborhood of the origin. This is obtained via an upper bound on $\norm{F^{-1}(c)}_1$ in terms of $\norm{c}_1$. The proof of the following lemma is given in \cref{sec:pf_inverse}. Since $F$ is not an injection globally, the $F^{-1}$ we refer to only means a continuous function whose domain is an open subset of $\RR^{\infty}$ containing 0 and satisfies $F\circ F^{-1} = \id$.
        
        \begin{lemma}[Invertibility of $F$ in $\RR^{\infty}$]\label{lm:inverse}
                Define
                \begin{equation}\label{eq:dfn_H}
                        H(x) := \int_0^x 2 - h(t)\rd t = 4x - \sinh(2x),
                \end{equation}
                \begin{equation}\label{eq:dfn_r1}
                r_{\Phi}:=h^{-1}(2)\approx 0.658,
        \end{equation}
                and 
                \begin{equation}\label{eq:dfn_r2}
                        r_c := H(r_{\Phi})  \approx 0.902.
                \end{equation}
                $F$ has an inverse map $F^{-1}: B(0,r_c)\subset \RR^{\infty}\to \mathbb{R}^{\infty}$. Moreover, for any $c\in  \RR^{\infty}$ with $\norm{c}_1<r_c$, it holds that
                \begin{equation}\label{eq:F_inverse}
                        \norm{F^{-1}(c)}_1\le H^{-1}(\norm{c}_1).
                \end{equation}
        \end{lemma}
        As a remark, the effective length of $F^{-1}(c)$ is always equal to that of $c$, which is implied by the proof of \cref{lm:inverse}.

\begin{corollary}\label{lm:equiv_dist_lower}
For any $c^{(j)}\in \RR^{\infty}$ with $\norm{c^{(j)}}_1\le \theta<r_c$ where $j=1,2$, define $\Phi^{(j)}:= F^{-1}(c^{(j)})$ for $j=1,2$. It holds that
\begin{equation}
    \wt{C}(\theta)\norm{\Phi^{(1)}-\Phi^{(2)}}_1 \le \norm{c^{(1)} - c^{(2)}}_1,
\end{equation}
where $\wt{C}(\theta)=2-h(H^{-1}(\theta))$.
\end{corollary}
\begin{proof}
First, one has
\begin{equation}
\begin{split}
    &\norm{c^{(1)}-c^{(2)}}_1= \norm{F(\Phi^{(1)})-F(\Phi^{(2)})}_1\\
    &=\norm{\int_{0}^1 DF\left(s\Phi^{(1)}+(1-s)\Phi^{(2)}\right)\cdot\left(\Phi^{(1)}-\Phi^{(2)}\right)\rd s}_1\\
    &=  \norm{\int_{0}^1 \left(2\bI +DF\left(s\Phi^{(1)}+(1-s)\Phi^{(2)}\right)-2\bI\right)\cdot\left(\Phi^{(1)}-\Phi^{(2)}\right)\rd s}_1\\
    &\geq 2\norm{\Phi^{(1)}-\Phi^{(2)}}_1 - \int_{0}^1 \norm{DF\left(s\Phi^{(1)}+(1-s)\Phi^{(2)}\right)-2\bI}_1 \norm{\Phi^{(1)}-\Phi^{(2)}}_1\rd s
\end{split}
\end{equation}
Hence, there exists $t\in (0,1)$ such that $\Phi'=t\Phi^{(1)}+(1-t) \Phi^{(2)}$ and 
\begin{equation}
    \norm{c^{(1)}-c^{(2)}}_1\geq2\norm{\Phi^{(1)}-\Phi^{(2)}}_1 - \norm{DF(\Phi')-2\bI}_1\norm{\Phi^{(1)}-\Phi^{(2)}}_1.
\end{equation}
From \cref{lm:inverse}, for $j=1,2$, one has
\begin{equation}
\norm{\Phi^{(j)}}_1=\norm{F^{-1}(c^{(j)})}_1\leq H^{-1}(\norm{c^{(j)}}_1)\leq H^{-1}(\theta)<r_\Phi.
\end{equation}
Then $\norm{\Phi'}_1\leq H^{-1}(\theta)$ follows by the convexity of 1-norm ball. According to \cref{lm:DF-2I_estimate}, one has 
$\norm{DF(\Phi')-2\bI}_1\leq  h(\norm{\Phi'}_1)<1$. Thus, we obtain 
\begin{equation}
\wt{C}(\theta)\norm{\Phi^{(1)}-\Phi^{(2)}}_1\leq \norm{c^{(1)}-c^{(2)}}_1,
\end{equation}
where $\wt{C}(\theta):=2-h( H^{-1}(\theta))$.
\end{proof}

 Combining \cref{re:equiv_dist_upper} and \cref{lm:equiv_dist_lower}, we obtain the following theorem.
\begin{lemma}[Equivalence of distance in $\RR^{\infty}$]\label{lm:equiv_dist}
For any $c^{(j)}\in \RR^{\infty}$ with $\norm{c^{(j)}}_1\le \theta<r_c$, $j=1,2$, define $\Phi^{(j)}:= F^{-1}(c^{(j)})$, $j=1,2$. It holds that
\begin{equation}
    \wt{C}(\theta)\norm{\Phi^{(1)}-\Phi^{(2)}}_1 \le \norm{c^{(1)} - c^{(2)}}_1\leq C(\theta)\norm{\Phi^{(1)}-\Phi^{(2)}}_1,
\end{equation}
where $\wt{C}(\theta)=2-h(H^{-1}(\theta))$ and $C(\theta)=C_1(H^{-1}(\theta))$.
\end{lemma}
An immediate consequence of \cref{lm:equiv_dist} is that: for any $c\in \RR^{\infty}$ with $\norm{c}_1\le \theta<r_c$, if we let $\Phi:= F^{-1}(c)$, then
\begin{equation}
    \wt{C}(\theta)\norm{\Phi}_1 \le \norm{c}_1\leq C(\theta)\norm{\Phi}_1.
\end{equation}
To show the sharpness of $r_c$, we plot $\wt{C}(\theta)$ and $C(\theta)$ as function of $\theta$ in \cref{fig:constant_visual}.
\begin{figure}[htbp]
        \centering
        \includegraphics[width =0.5\textwidth]{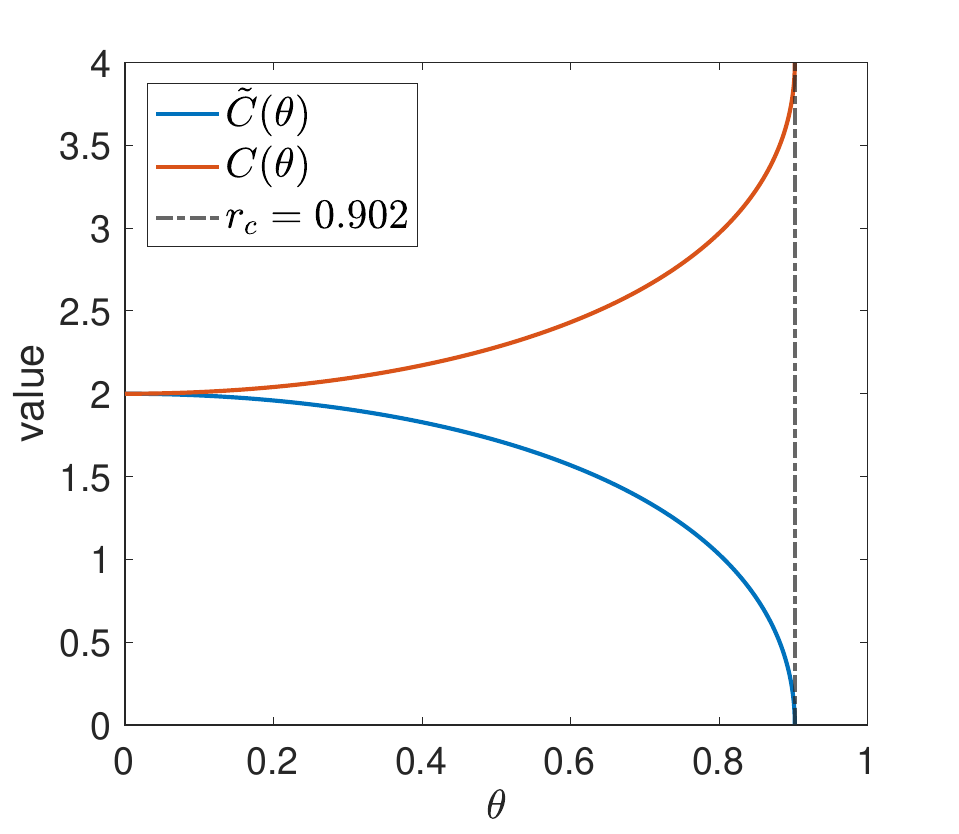}
        \caption{The plot of $\wt{C}(\theta)$ and $C(\theta)$ as function of $\theta$.}\label{fig:constant_visual}
\end{figure}

According to \cref{lm:equiv_dist}, the map $F$ preserves the distance up to a constant, i.e., $F$ is a \emph{quasi isometry} near the origin.

\subsection{Extension to $\ell^1$}\label{sec:extension}

In this section, we extend \cref{lm:inverse} and \cref{lm:equiv_dist} from $\RR^\infty$ to $\ell^1$. The map $\mathcal{F}$ is a well-defined mapping from  $S_e$ (or $S_o$) to $\ell^1$ according to
\begin{equation}
\mathcal{F}(f)=(c_0,c_1,\cdots).
\end{equation}
The following theorem is a more detailed statement of \cref{thm:main_l1}.   
    \begin{theorem}[Invertibility of $F$ in $\ell^1$]\label{thm:inverse_l1}
        The map $F:\RR^\infty\to \RR^\infty$ can be extended to $\overline{F}:\ell^1\to \ell^1$. Furthermore, $\overline{F}$ has an inverse map $\overline{F}^{-1}: B(0,r_c)\subset \ell^1\to \ell^1$. For any $c\in  \RR^{\infty}$ with $\norm{c}_1<r_c$, it holds that
        \begin{equation}
            \norm{\overline{F}^{-1}(c)}_1\le H^{-1}(\norm{c}_1).
        \end{equation}
    \end{theorem}
\begin{proof}
\cref{lm:DF_lipschitz} and \cref{re:equiv_dist_upper} state that $F$ and $DF$ are both Lipschitz continuous. By \cref{thm:extend} and the fact that $\RR^{\infty}$ is a dense subspace of $\ell^1$, $F$ can be extended to the whole $\ell^1$. 
The inverse mapping theorem for Banach spaces (\cref{thm:inverse mapping}), together with \cref{thm:extend} imply that $\overline{F}$ shares the same property with $F$ within a neighborhood of the origin. This completes the proof due to \cref{lm:inverse}.
\end{proof}

Following the proof of \cref{thm:inverse_l1}, we can also extend \cref{lm:equiv_dist} to $\ell^1$, which states that $\overline{F}$  preserves the distance up to a constant in a neighborhood of the origin.

\begin{theorem} [Equivalence of distance in $\ell^1$]\label{lm:equiv_dist_l1}
For any $c^{(j)}\in \ell^1$ with $\norm{c^{(j)}}_1\le \theta<r_c$ where $j=1,2$, define $\Phi^{(j)}:= \overline{F}^{-1}(c^{(j)})$, $j=1,2$. It holds that
\begin{equation}
    \wt{C}(\theta)\norm{\Phi^{(1)}-\Phi^{(2)}}_1 \le \norm{c^{(1)} - c^{(2)}}_1\leq C(\theta)\norm{\Phi^{(1)}-\Phi^{(2)}}_1,
\end{equation}
where $\wt{C}(\theta)=2-h(H^{-1}(\theta))$ and $C(\theta)=C_1(H^{-1}(\theta))$.
\end{theorem}

With the help of the inverse mapping theorem of Banach spaces, the proof of this theorem follows the same idea as \cref{lm:equiv_dist}.

Now we are ready to give a positive answer to the first question raised in \cref{sec:background}. 
Let $\Psi^{(d)}$ denote the symmetric phase factors and $\Phi^{(d)}$ denote the corresponding reduced phase factors. Although the solution to  $F(\Phi)=\mathcal{F}(f^{(d)})$ may not be unique, \cref{thm:inverse_l1} allows us to choose one $\Phi^{(d)}$ such that $\norm{\Phi^{(d)}}_1<r_{\Phi}$ as long as $\norm{\mathcal{F}(f^{(d)})}_1<r_c$. Assume that $c:=\mathcal{F}(f)$ satisfies $\norm{c}_1<r_c$ as well. 
Then $\lim_{d\to \infty} f^{(d)} = f$ implies that $\{\mathcal{F}(f^{(d)})\}$ converges to $c$ with respect to the vector $1$-norm. The convergence of $\{\Phi^{(d)}\}$ to $\overline{F}^{-1}(c)$ follows by applying the equivalence of distance in \cref{lm:equiv_dist_l1}.
\cref{thm:inverse_l1} provides a sufficient condition to the existence of the solution to $\overline{F}(\Phi)=c$. We would like to emphasize that similar to the case of polynomials, the solution might not be unique in $\ell^1$. 

\subsection{Structure of $DF$}
 In this section, we focus on the structure of the Jacobian matrix $DF$. Given any $\Phi\in \RR^{\infty}$, we let $n$ be its effective length. According to \cref{eqn:DF_column}, the $k$-th column of $DF(\Phi)$, i.e., $\frac{\partial F(\Phi)}{\partial \phi_k}$, is the Chebyshev-coefficient vector of $\frac{\partial g(x,\Phi)}{\partial \phi_k}$. Direct computation shows
        \begin{equation}\label{eq:partial_F}
            \frac{\partial F(\Phi)}{\partial \phi_k} = F(\Phi+\frac{\pi}{4}e_k)-F(\Phi-\frac{\pi}{4}e_k), \quad \forall k\in \NN, \Phi\in\RR^{\infty}.
        \end{equation}
Here, $e_k\in\RR^{\infty}$ is the vector with all components equal to $0$ except for the $k$-th component which is equal to $1$. As a reminder, when we refer to the component of a vector or matrix, the index begins with 0. 
        
        We consider the even case for simplicity. For $k< n$, according to \cref{eq:partial_F}, $\frac{\partial g(x,\Phi)}{\partial \phi_k}$ is a polynomial of degree at most $2n-2$. Thus the components $\left[\frac{\partial F(\Phi)}{\partial \phi_k}\right]_j=0$ for any $j\ge n$. For $k\geq n$,  $\frac{\partial g(x,\Phi)}{\partial \phi_k}$ is a polynomial of degree at most $2k$, and then $\left[\frac{\partial F(\Phi)}{\partial \phi_k}\right]_j=0$  for any $j>k$. 
Therefore $DF(\Phi)$ takes the form
        \begin{equation}
            \left(\begin{matrix}
                  D_1 &D_2\\
                    0&D_3
            \end{matrix}\right).
        \end{equation}
        Here, $D_1$ is a matrix of size $n\times n$ and $D_3$ is an upper triangular matrix of infinite dimension. As for matrix $D_2$, the number of rows is $n$, while the number of columns is infinite.
        
When $DF(\Phi)$ is invertible, the inverse takes the form
        \begin{equation}
            DF(\Phi)^{-1}=\left(\begin{matrix}
                 D_1^{-1}&-D_1^{-1} D_2 D_3^{-1}\\
                    0&D_3^{-1}
            \end{matrix}\right).
        \end{equation}
After extending $F$ to $\ell^1$, we can characterize the invertibility of matrix $D\overline{F}$ in $\ell^1$. 

    \begin{lemma}\label{lm:DF_invertible}
        For any $\Phi\in \ell^1$ with $\norm{\Phi}_1<r_{\Phi}$, $D\overline{F}(\Phi)$ is invertible and 
        \begin{equation}
            \norm{D\overline{F}(\Phi)^{-1}}_1\leq \frac{1}{2-h(\norm{\Phi}_1)}.
        \end{equation}
    \end{lemma}
    
    \begin{proof}
        When $\norm{\Phi}_1<h^{-1}(2)$, we have $\norm{\half D\overline{F}(\Phi)-\bI}_1\leq \frac{1}{2} h(\norm{\Phi}_1)<1$. Applying \cref{thm:invertible} in the Banach space $\ell^1$, $D\overline{F}(\Phi)$ is invertible and 
        \begin{equation*}
            \norm{D\overline{F}(\Phi)^{-1}}_1 = \half\norm{\left(\half D\overline{F}(\Phi)\right)^{-1}}_1\leq \frac{\half}{1-\frac{1}{2}h(\norm{\Phi}_1)}=\frac{1}{2-h(\norm{\Phi}_1)}.
        \end{equation*}
    \end{proof}

\section{Decay properties}\label{sec:decay}

As an immediate consequence of \cref{lm:equiv_dist_l1}, we now prove the decay properties of the reduced phase factors $\Phi\in\ell^1$ (\cref{thm:main_decay}) with an explicit characterization of the constant $C$.

\begin{proof}[Proof of \cref{thm:main_decay}]
Let $\Phi^{(n)} = (\phi_0,\ldots, \phi_n,0,0,\ldots)$, and $\theta = \norm{c}_1$ in \cref{lm:equiv_dist_l1}, then we get
    \begin{equation}
    \begin{split}
            C_1(H^{-1}(\norm{c}_1))\sum_{k>n}|\phi_k|&=C_1(H^{-1}(\norm{c}_1))\norm{\Phi-\Phi^{(n)}}_1\\
            &\ge\norm{\overline{F}(\Phi)-\overline{F}\left(\Phi^{(n)}\right)}_1=\norm{c-\overline{F}\left(\Phi^{(n)}\right)}_1\\
            &\ge\sum_{k>n}\abs{c_k},
    \end{split}
    \end{equation}
    where the last inequality follows the fact that $\overline{F}(\Phi^{(n)})$ is zero after the $n$-th component. 
    Similarly, we choose $c^{(n)} = (c_0,\ldots, c_n,0,0,\ldots)$, and \cref{lm:equiv_dist_l1} gives
    \begin{equation}
    \begin{split}
            \frac{1}{2-h(H^{-1}(\norm{c}_1))}\sum_{k>n}|c_k|&=\frac{1}{2-h(H^{-1}(\norm{c}_1))}\norm{c-c^{(n)}}_1\\
            &\ge\norm{\overline{F}^{-1}(c)-\overline{F}^{-1}\left(c^{(n)}\right)}_1=\norm{\Phi-\overline{F}^{-1}\left(c^{(n)}\right)}_1\\
            &\ge\sum_{k>n}\abs{\phi_k},
    \end{split}
    \end{equation}
    where the last inequality follows the fact that $\overline{F}^{-1}(c^{(n)})$ is zero after the $n$-th component. This completes the proof with $C'=\frac{1}{C_1(H^{-1}(\norm{c}_1))}$ and $C=\frac{1}{2-h(H^{-1}(\norm{c}_1))}$.
\end{proof}
If $\sum_{k>n} \abs{c_k}=O(n^{-\alpha})$ for some $\alpha >0$, then $\sum_{k>n}\abs{\phi_k}=O(n^{-\alpha})$. If $c_k$ decays super-algebraically or exponentially, so does the reduced phase factors. Therefore the tail decay of the reduced phase factors is determined by the smoothness of the target function.

\section{Fixed-point iteration for finding phase factors}\label{sec:algorithm}

\cref{alg:iterative-qsp} is a very simple iterative algorithm  based on fixed-point iteration for finding phase factors. Numerical results suggest that this algorithm is quite robust, starting from a \emph{fixed} initial point $\Phi^{0} = 0$ (or $\Phi^{1} = \half c$ according to \cref{alg:iterative-qsp}). 
We emphasize that the choice of the initial guess is important, and other initial points may make the algorithm diverge. 
Based on the developments in \cref{sec:inf_sym_QSP}, we prove that 
\cref{alg:iterative-qsp} converges linearly in $\ell^1$, and describe the computational complexity.

\subsection{Convergence}\label{sec:convergence}
To prove the convergence of \cref{alg:iterative-qsp}, it is sufficient to prove that 
        \begin{equation}
                G(\Phi):=\Phi-\frac{1}{2} F\left(\Phi\right)+\half F(\Phi^{\star})
        \end{equation}
        is a contraction map in a neighborhood of $\Phi^{\star}$. $\Phi^{\star}$ denotes the desired set of reduced phase factors. The contraction property of $G$ follows the observation that the Jacobian matrix $DG(\Phi^{\star}) = \bI - \half DF(\Phi^{\star}) = \half\left(DF(0) - DF(\Phi^{\star})\right)$ would be small when $\norm{\Phi^{\star}}_1$ is sufficiently small according to \cref{lm:DF-2I_estimate}. We define a function 
        \begin{equation}
            \gamma(r):=\frac{1}{2}\int_{0}^{1} h\left( r+s\left(\frac{1}{2}\sinh(2r)-r\right)\right)\rd s,
        \end{equation}
        and the following constants
        \begin{align}
                \wt{r}_{\Phi} &:= \half\arcsinh(\arccosh(2))\approx 0.544\\
                \wt{r}_c&:= H(\wt{r}_{\Phi})\approx 0.861\\
                \wt{\gamma}&:=\gamma(\wt{r}_{\Phi})\approx 0.8189.
        \end{align}
        We will also use the fact that  $\wt{r}_{\Phi}<\half\arccosh(2) = r_{\Phi} \approx 0.658$.

\begin{lemma}\label{lm:cm}
    If $\Phi^{\star}$ satisfies $c = \overline{F}(\Phi^{\star})$ and $\norm{\Phi^{\star}}_1\leq\wt{r}_{\Phi}$, $G$ is a contraction map in the open ball  $B:=B\left(\Phi^{\star}, \norm{\half F(\Phi^{\star})-\Phi^{\star}}_1\right)$. 
\end{lemma}
\begin{proof}
First, we give an upper bound on the radius of ball $B$.
        \begin{equation}\label{eqn:radius}
             \begin{aligned}
                \norm{\half F(\Phi^{\star})-\Phi^{\star}}_1&=\half\norm{\int_0^1 \frac{\rd}{\rd t} \left(F(t\Phi^{\star})-2t\Phi^{\star}\right)\rd t}_1&\\
                &= \half\norm{\int_0^1 \left(DF(t\Phi^{\star})-2\bI\right)\cdot \Phi^{\star}\rd t}_1&\\
                &\le\half\int_0^{1}\norm{DF\left(t\Phi^{\star}\right)-2\bI}_1\norm{\Phi^{\star}}_1\rd t&\text{(use \cref{lm:DF-2I_estimate})}\\
                &\le\half\int_0^{\norm{\Phi^{\star}}_1}h(t)\rd t= \half\sinh(2\norm{\Phi^{\star}}_1)-\norm{\Phi^{\star}}_1.&
                \end{aligned}
       \end{equation}
                Thus, for any $\Phi\in B$, we have the following estimate
                \begin{equation}\label{eqn:estimate_phi}
                \begin{split}
                        \norm{\Phi}_1&< \norm{\Phi^{\star}}_1+\norm{\half F(\Phi^{\star})-\Phi^{\star}}_1\le\half \sinh(2\norm{\Phi^{\star}}_1)\\
                        &\leq\half \sinh(2\wt{r}_{\Phi}) = \half\arccosh(2) = r_{\Phi}. 
                \end{split}
                \end{equation}
                Then we can use \cref{lm:DF-2I_estimate} again to conclude
                \begin{equation}
                        \norm{DG(\Phi)}_1 = \norm{\bI-\half DF(\Phi)}_1\le \half h(\norm{\Phi}_1)< \half h(r_{\Phi}) = 1.
                \end{equation}
                This means $G$ is a contraction map in the ball $B$ and $\Phi^{\star}$ is its only fixed point.
\end{proof}
        \begin{lemma}\label{lm:contraction}
            If $\Phi^{\star}$ satisfies $c = \overline{F}(\Phi^{\star})$ and $\norm{\Phi^{\star}}_1\leq\wt{r}_{\Phi}$, \cref{alg:iterative-qsp} converges  $Q$-linearly to $\Phi^{\star}$. The rate of convergence is bounded by $\frac{1}{2}h(\norm{\Phi^{\star}}_1)$, i.e.,
            \begin{equation}
                \lim_{t\to \infty} \frac{\norm{\Phi^{t+1}-\Phi^{\star}}_1}{\norm{\Phi^{t}-\Phi^{\star}}_1}\le \frac{1}{2}h(\norm{\Phi^{\star}}_1).
                \end{equation}
            Here $\Phi^t$ is the set of reduced phase factors in the $t$-th iteration step.
            Furthermore, for every $t\ge 1$, 
                \begin{equation}
                        \frac{\norm{\Phi^{t+1}-\Phi^{\star}}_1}{\norm{\Phi^{t}-\Phi^{\star}}_1}\le \gamma(\norm{\Phi^{\star}}_1).
                \end{equation}
                In particular, $\gamma(\norm{\Phi^{\star}}_1)\leq \wt{\gamma}\approx 0.8189$. 
        \end{lemma}
        \begin{proof}
                 \cref{lm:cm} guarantees the convergence of \cref{alg:iterative-qsp} as long as $\Phi^{t}\in B$. However, we note that $\Phi^{1} = \half F(\Phi^{\star}) $ lies on the boundary of $B$. Hence, we need a finer estimation.
                \begin{equation}
                    \begin{split}
                        \norm{\Phi^{t+1}-\Phi^{\star}}_1&= \norm{G(\Phi^{t})-G(\Phi^{\star})}_1\\
                        &=\norm{\int_{0}^{1} DG\left(s\Phi^{t}+(1-s)\Phi^{\star}\right)\cdot\left(\Phi^{t}-\Phi^{\star}\right)\rd s}_1\\
                        &\leq \int_{0}^{1} \norm{DG\left(s\Phi^{t}+(1-s)\Phi^{\star}\right)}_1\norm{\Phi^{t}-\Phi^{\star}}_1\rd s\\
                        &\leq \frac{1}{2}\norm{\Phi^{t}-\Phi^{\star}}_1
                        \int_{0}^{1} h\left(\norm{s\Phi^{t}+(1-s)\Phi^{\star}}_1\right)\rd s\\
                        &\leq\frac{1}{2} \norm{\Phi^{t}-\Phi^{\star}}_1
                        \int_{0}^{1} h\left(s\norm{\Phi^{t}}_1+(1-s)\norm{\Phi^{\star}}_1\right)\rd s
                    \end{split}
                \end{equation}
                The last inequality follows that $h$ is monotonic increasing on $[0,\infty)$. 
                
               By replacing the ``$<$'' by ``$\leq$" in \cref{eqn:estimate_phi}, we get $\norm{\Phi^{1}}_1\leq r_\Phi$. Hence, 
                \begin{equation}
                    \begin{split}
                       \norm{\Phi^{2}-\Phi^{\star}}_1&\leq \frac{1}{2}\norm{\Phi^{1}-\Phi^{\star}}_1
                        \int_{0}^{1} h\left(s r_\Phi+(1-s)\wt{r}_\Phi\right)\rd s\\
                        &\leq \wt{\gamma} \norm{\Phi^{1}-\Phi^{\star}}_1.
                    \end{split}
                \end{equation}
               Note that $\wt{\gamma}\approx 0.8189$ implies that $\Phi^{2}\in B$. According to \cref{lm:cm}, we know that the rate of convergence is bounded by $\norm{DG(\Phi^{\star})}_1\leq \frac{1}{2}h(\norm{\Phi^{\star}}_1)$.
                
                Furthermore, for any $t\ge 1$ it holds that
                \begin{equation}
                    \begin{split}
                       \norm{\Phi^{t+1}-\Phi^{\star}}_1&\leq \frac{1}{2}\norm{\Phi^{t}-\Phi^{\star}}_1
                        \int_{0}^{1} h\left( \frac{s}{2}\sinh(2\norm{\Phi^{\star}}_1)+(1-s)\norm{\Phi^{\star}}_1\right)\rd s\\
                        & = \norm{\Phi^{t}-\Phi^{\star}}_1 \gamma(\norm{\Phi^{\star}}_1).
                    \end{split}
                \end{equation}
                Here, we use the result that for any $\Phi\in B$, $\norm{\Phi}_1\leq \frac{1}{2}\sinh(2\norm{\Phi^{\star}}_1)$ from \cref{eqn:estimate_phi}. This proves the lemma.
        \end{proof}
        As a remark, both the theoretical analysis and numerical results suggest that the region in which \cref{alg:iterative-qsp} converges should be larger than $B$. 
We now prove \cref{thm:main_converge} (i) with an explicit characterization of the constant $C$.

        \begin{proof}[Proof of \cref{thm:main_converge} (i)]
                \cref{eq:F_inverse} implies that $\Phi^{\star} := F^{-1}(c)$ satisfies $\norm{\Phi^{\star}}_1<\wt{r}_{\Phi}$, given $\norm{c}_1 < H(\wt{r}_{\Phi})$.  According to \cref{eqn:radius}, we get 
                \begin{equation}
                \begin{split}
                \norm{\Phi^{1}-\Phi^{\star}}_1&=\norm{\frac{1}{2} F(\Phi^{\star})-\Phi^{\star}}_1 \leq \half\sinh(2\norm{\Phi^{\star}}_1)-\norm{\Phi^{\star}}_1\\
                &\leq \half \sinh(2\wt{r}_{\Phi}) -\wt{r}_{\Phi}= r_{\Phi}-\wt{r}_{\Phi}.
                \end{split}
                \end{equation}
                Here we use the fact that $\half \sinh(2x) -x$ is monotonic increasing on $[0,\infty)$. Applying \cref{lm:contraction}, we obtain the error estimate 
\begin{equation}\label{eqn:error}
 \norm{\Phi^{t}-\Phi^{\star}}_1 \le (r_{\Phi}-\wt{r}_{\Phi})\wt{\gamma}^{t-1}, \quad \wt{\gamma} \approx 0.8189 ,\quad t\ge 1.
 \end{equation}
This finishes the proof of \cref{thm:main_converge} (i), and the constant $C$ is $(r_{\Phi}-\wt{r}_{\Phi})/\wt{\gamma}\approx 0.1393$.
        \end{proof}

\subsection{Complexity}

In this subsection, we discuss the complexity of \cref{alg:iterative-qsp}.
For any target function satisfying $\norm{c}_1 < \wt{r}_c$, to reach the $\ell^1$-error tolerance $\epsilon$, the number of iterations is at most
\begin{equation}
    \left\lceil\log\left(\frac{1}{(r_\Phi-\wt{r}_\Phi)\epsilon}\right)/\log\wt{\gamma}\right\rceil.
\end{equation}
That means the number of iterations is
$\mathcal{O}(\log\frac{1}{\epsilon})$, and the upper bound of the number of iterations is independent of the target function and the effective length of $\Phi$. Therefore, we only need to analyze the complexity of implementing $F$. 
On a computer we can only perform operations for matrices of finite sizes. For any set of phase reduced phase factors $\Phi$ whose effective length is $\wt{d}$, let $\mathfrak{G}(\Phi)=\left(g\left(x_{0}, \Phi\right), \cdots, g\left(x_{2d}, \Phi\right)\right)^{\top}$. Here, $d=2\wt{d}-2$ if the target polynomial is even, and $d=2\wt{d}-1$ if odd. The Chebyshev node is $x_j= \cos(\frac{2\pi j}{2d+1})$.

Observe that for $k,l=0,\cdots,d$,  
\begin{equation}
\begin{split}
    \Re\left(\sum_{j=0}^{2d} T_{k}(x_j)e^{-\I \frac{2\pi}{2d+1}l j}\right)&=\sum_{j=0}^{2d}\cos\left(\frac{2\pi k}{2d+1}j\right)\cos\left(\frac{2\pi l}{2d+1}j\right)\\
    &=\frac{1}{2}\sum_{j=0}^{2d}\cos\left(\frac{2\pi (k+l)}{2d+1}j\right)+ \frac{1}{2}\sum_{j=0}^{2d}\cos\left(\frac{2\pi (k-l)}{2d+1}j\right)\\
    & = \frac{2d+1}{2}\delta_{kl}(\delta_{k0}+1).
\end{split}
\end{equation}
Hence the Chebyshev-coefficient vector can be evaluated by applying the fast Fourier transform (FFT) to $\mathfrak{G}(\Phi)$. 
The output from applying FFT to $\mathfrak{G}(\Phi)$ is a vector of length $2d+1$, denoted as $v$, satisfying 
\begin{equation}
    \Re(v_0)= (2d+1) \wt{c}_0, \quad \Re(v_j) = \frac{2d+1}{2} \wt{c}_j, j=1,\cdots, d
\end{equation}
where $\wt{c}_j$, $j=0,\cdots, d$, are the Chebyshev coefficients of $g(x,\Phi)$ with respect to $T_j$. Recall that $F(\Phi)$ is the Chebyshev-coefficient vector of $g(x,\Phi)$ which is either $(\wt{c}_0,\wt{c}_2,\cdots,\wt{c}_d)$ or $(\wt{c}_1,\wt{c}_3,\cdots,\wt{c}_d)$ depending on the parity of $d$. 

For completeness, the procedure for computing $F(\Phi)$ is given in \cref{alg:fft}. 
The cost of evaluating $\mathfrak{G}(\Phi)$ is $\Or(d^2)$, and the cost of FFT is $\mathcal{O}(d \log d)$.   Therefore, the overall time complexity is $\mathcal{O}(d^2\log\frac{1}{\epsilon})$. This concludes the proof of \cref{thm:main_converge} (ii).

\cref{alg:iterative-qsp} is numerically stable, in the sense that the number of bits required in the computation is $\Or(\mathrm{polylog}(d/\epsilon))$\footnote{Generically, the number of bits cannot be rigorously bounded by a constant in a proof. Most numerically stable algorithms can be robustly implemented using fixed double precision arithmetic operations in practice.}. To show the numerical stability, we consider the model of finite precision arithmetic (see standard axioms in \cite{Higham}) and denote by $u$ the unit roundoff error. In each iteration step, the rounding error when evaluating $\mathfrak{G}(\Phi)$ is $\Or(d^2u)$. The rounding error occurred in FFT is $\Or(\log(d) u)$ \cite{Ramos}. So the total rounding error accumulated in each iteration step remains $\mathcal{O}(d^2 u)$. Therefore the number of bits required by \cref{alg:iterative-qsp} is $\mathcal{O}\left(\log\left(\frac{d\log(1/\epsilon)}{\epsilon}\right)\right)$.

\begin{algorithm}[htbp]
\caption{Compute $F(\Phi)$.}
\label{alg:fft}
\begin{algorithmic}
\STATE{\textbf{Input:} Reduced phase factors $\Phi$, parity, and its effective length $\wt{d}$.}
\IF{parity is even}
\STATE{Set $d=2\wt{d}-2$.}
\ELSE
\STATE{Set $d=2\wt{d}-1$.}
\ENDIF
\STATE{Initialize $\vg = (0, 0, \cdots) \in \RR^{2d+1}$.}
\STATE{Evaluate $\vg_j \leftarrow g(x_j, \Phi),x_j=\frac{2\pi j}{2d+1},j = 0, \cdots, 2d+1$.}
\STATE{Compute $v_l\leftarrow \Re\left(\sum_{j=0}^{2d-1} \vg_j e^{-\I \frac{2\pi}{2d+1}l j}\right),l=0,\ldots,d$ using FFT.}
\IF{parity is even}
\STATE{$F(\Phi)\leftarrow \frac{2}{2d+1}(\frac{v_0}{2}, v_2, v_4,\cdots, v_d)$.}
\ELSE
\STATE{$F(\Phi)\leftarrow \frac{2}{2d+1}(v_1, v_3, v_5,\cdots, v_d)$.}
\ENDIF
\STATE{\textbf{Output:} $F(\Phi)$.}
\end{algorithmic}
\end{algorithm}

\section{Numerical results}\label{sec:numerics}
We present a number of tests to demonstrate the efficiency of the fixed-point iteration (FPI) algorithm (\cref{alg:iterative-qsp}). All numerical tests are performed on a 6-core Intel Core i7 processor at 2.60 GHz with 16 GB of RAM. Our method is implemented in \textsf{MATLAB} R2019a.  

The target function is $f(x)=e^{-\I\tau x}$, which has applications in Hamiltonian simulation. It can be expanded using the Jacobi-Anger expansion as
\begin{equation}
\label{eq;Jacobi-Anger}
    e^{-\I\tau x}=J_0(\tau)+2\sum_{k\text{ even}}(-1)^{k/2} J_k(\tau) T_k(x) + 2\I \sum_{k \text{ odd}} (-1)^{(k-1)/2} J_k(\tau) T_k(x),
\end{equation}
where $J_k$'s are the Bessel functions of the first kind. 

We use \cref{alg:iterative-qsp} to find the phase factors respectively for the even part
\begin{equation}
    f_{\text{even}}(x):= J_0(\tau)+2\sum_{k\text{ even}, k<d}(-1)^{k/2} J_k(\tau) T_k(x),
\end{equation}
and the odd part 
\begin{equation}
    f_{\text{odd}}(x):= 2\sum_{k \text{ odd}, k<d} (-1)^{(k-1)/2} J_k(\tau) T_k(x).
\end{equation}
We choose $d=1.4\abs{\tau}+\log(1/\epsilon_0)$, where $\epsilon_0 = 10^{-14}$. Since the target polynomial should be bounded by 1, to ensure numerical stability, we scale $f_{\text{even}}(x)$ and $f_{\text{odd}}(x)$ by a factor of $\half$. Then we use \cref{alg:iterative-qsp} to find phase factors for target polynomials, $\half f_{\text{even}}(x)$ and $\half f_{\text{odd}}(x)$, respectively. \cref{fig:iter_res_HS} displays the corresponding residual error, $\norm{F(\Phi^{t})-c}_1$ with $\tau= 1000$, where $\Phi^{t}$ is the set of reduced phase factors at the $t$-th iteration.  Note that $\norm{c}_1$ is very large, and is equal to 9.8609 and 9.7403 for the even and odd case, respectively. 
Nonetheless, \cref{alg:iterative-qsp} converges starting from the fixed initial guess.

In \cref{fig:converge_HS}, we demonstrate that for a fixed target function, as the polynomial degree $d$ increases, the corresponding set of reduced phase factors $\Phi^{(d)}$ indeed converges to some $\Phi^{\star}$.  We still take $\half f_{\text{even}}(x)$ and $\half f_{\text{odd}}(x)$ as examples. We choose $\tau=200$, and $\Phi^{\star}$ are approximately computed by setting the degrees of $\half f_{\text{even}}(x)$ and $\half f_{\text{odd}}(x)$ to 312 and 313 respectively. We truncate $\half f_{\text{even}}(x)$ to even polynomials with degree $d=180,190,\cdots,310$. Similarly, we also truncate $\half f_{\text{odd}}(x)$ to odd polynomials with degree $d=181,191,\cdots,311$. Then we use \cref{alg:iterative-qsp} to compute the corresponding reduced phase factors, $\Phi^{(d)}$. We would like to emphasize that to our knowledge, only the optimization based method is able to find $\Phi^{(d)}$ that has a well defined limit as $d\to \infty$.

We also compare the performance of \cref{alg:iterative-qsp} with the quasi-Newton method implemented in \cite{DongMengWhaleyEtAl2021} on $\half f_{\text{even}}(x)$ and $\half f_{\text{odd}}(x)$ with $\tau=50, 100, \cdots, 1000$. The stopping criteria for \cref{alg:iterative-qsp} is  $\norm{F(\Phi^t)-c}_1\leq \epsilon$. As for quasi-Newton method, the iteration stops when
\begin{equation}
    L(\Phi):=\max_{j=1,\cdots,\wt{d}}\abs{ g(x_j,\Phi^t)-f_{\text{target}}(x_j)}\leq \epsilon.
\end{equation}
where $f_{\text{target}}$ is the target polynomial and $x_j = \cos\left(\frac{(2j-1)\pi}{4\wt{d}}\right)$, $j=1,\cdots, \wt{d}$ are the positive roots of the Chebyshev polynomial $T_{2\wt{d}}(x)$. For numerical demonstration, we choose $\epsilon=10^{-12}$. The results of comparison are displayed in \cref{fig:comparison}. Since $L(\Phi)\leq \norm{F(\Phi)-c}_1$ for any $\Phi$, the stopping criteria for \cref{alg:iterative-qsp} is actually tighter than that in quasi-Newton method. Thus \cref{fig:cpu_d_HS} implies that \cref{alg:iterative-qsp} converges faster than quasi-Newton method in this example. The degree of the target polynomial linearly increases as the value of $\tau$ increases, and the CPU time scales asymptotically as $\Or(\tau^2)=\Or(d^2)$. In \cref{fig:iter_d_HS}, we present the number of iterations required using FPI to find phase factors for $\half f_{\text{even}}(x) $ and $\half f_{\text{odd}}(x)$.  \cref{fig:iter_d_HS} indicates that the number of iterations is almost independent of the degree of the target polynomial.

Finally, we demonstrate the decay of phase factors in \cref{fig:numerical-pf-decay}. In the first example, we truncate the series expansion of $f(x)=0.8\abs{x}^3$ in terms of Chebyshev polynomials of the first kind up to degree $d=1000$ and use \cref{alg:iterative-qsp} to find the corresponding reduced phase factors. The third order derivative of $f(x)$ is discontinuous. In \cref{fig:decay_absx3}, we plot the magnitude of its Chebyshev-coefficient vector, as well as the reduced phase factors obtained by \cref{alg:iterative-qsp}.  Here, the 1-norm of Chebyshev coefficients is about $0.8149$, which is bounded by $r_c$. \cref{fig:decay_absx3} shows that the reduced phase factors decay away from the center with an algebraic decay rate around $4$, which matches the decay rate of Chebyshev coefficients. This also agrees with our theoretical results in \cref{thm:main_decay}. In the second example, we choose $\half f_{\text{odd}}(x)$ as target polynomial and present the magnitude of both Chebyshev-coefficient vector and the corresponding reduced phase factors in \cref{fig:decay_HS_odd}. The 1-norm of Chebyshev coefficients is around $3.2332$, which exceeds the norm condition in \cref{thm:main_decay}. 
Nonetheless, the decay of the tail of the phase factors closely matches that of the Chebyshev-coefficient vector.

\begin{figure}
\centering
\begin{subfigure}{0.45\textwidth}
  \includegraphics[width=60mm]{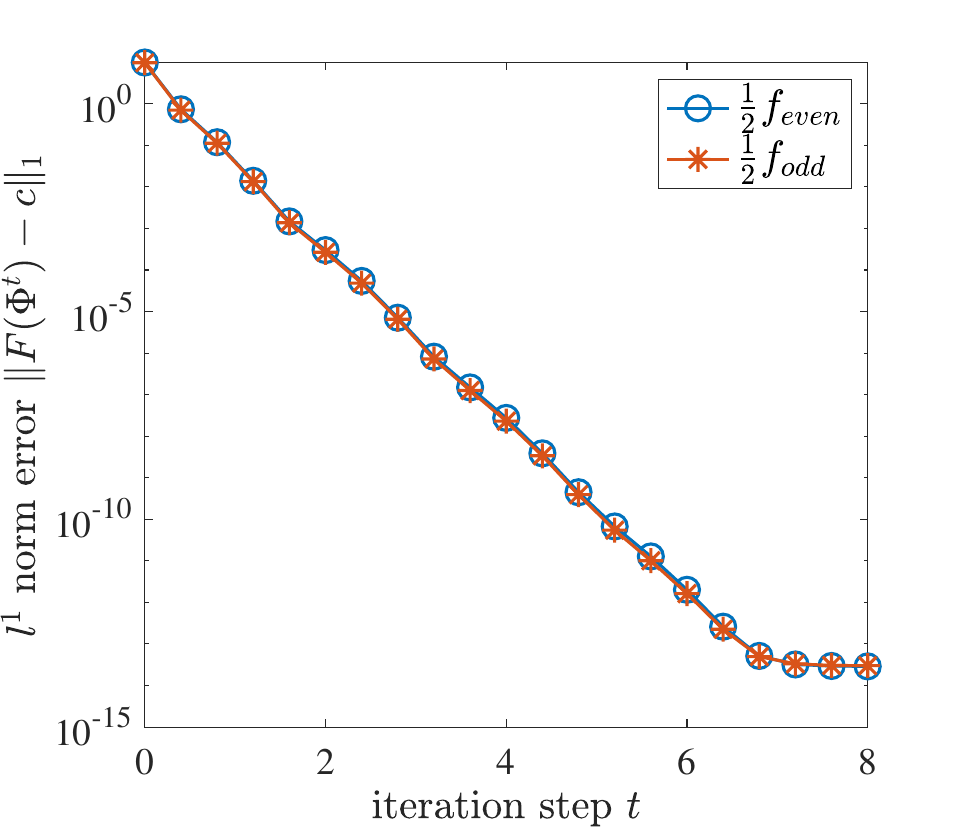}
  \caption{The residual error after each iteration by using \cref{alg:iterative-qsp} to determine phase factors for $\half f_{\text{even}}(x)$ and $\half f_{\text{odd}}(x)$ respectively, where $\tau=1000$.}
  \label{fig:iter_res_HS}
\end{subfigure}
\hfill
\begin{subfigure}{0.45\textwidth}
  \includegraphics[width=60mm]{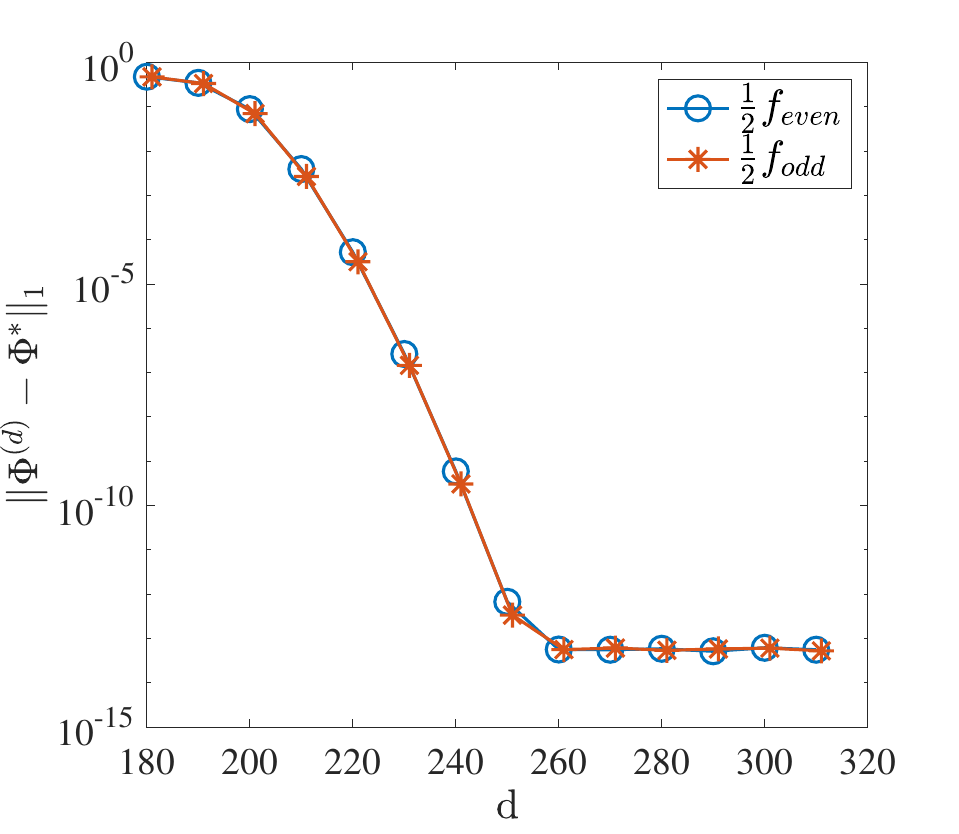}
  \caption{The convergence of $\Phi^{(d)}$ as the degree of polynomial approximating target function increases. The target functions are $\half f_{\text{even}}(x)$ and $\half f_{\text{odd}}(x)$ of degree 312 and 313 respectively, where $\tau =200$.}
  \label{fig:converge_HS}
\end{subfigure}
\caption{The performance of the fixed-point iteration (FPI) algorithm (\cref{alg:iterative-qsp}) to find phase factors for $\half f_{\text{even}}(x)$ and $\half f_{\text{odd}}(x)$.}
\end{figure}

\begin{figure}
\centering
\begin{subfigure}{0.45\textwidth}
  \includegraphics[width=60mm]{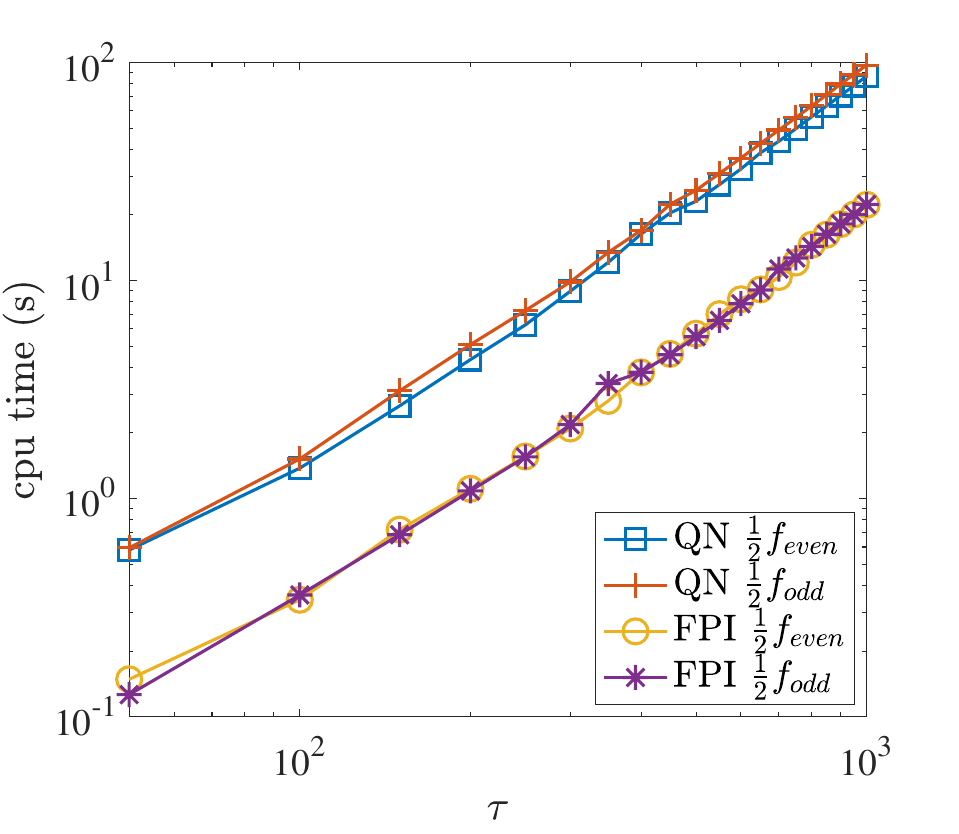}
  \caption{CPU time(s) required using FPI and  \textsf{QSPPACK} to find QSP phase factors. The slopes of the red and purple lines are about 2, representing CPU time = const$\cdot \tau^2$.}
  \label{fig:cpu_d_HS}
\end{subfigure}
\hfill
\begin{subfigure}{0.45\textwidth}
  \includegraphics[width=60mm]{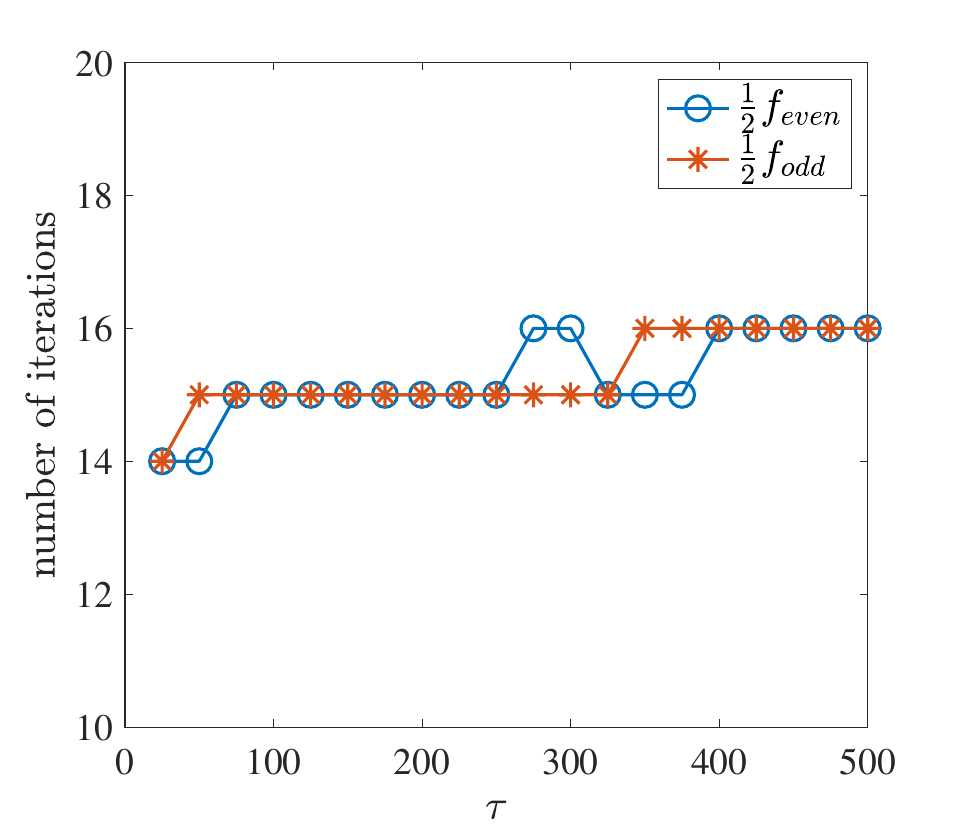}
  \caption{Number of iterations required using FPI to find QSP phase factors. The number of iterations required varies between 14 and 16 and has no direct relation to $\tau$.}
  \label{fig:iter_d_HS}
\end{subfigure}
\caption{Comparison of the performance of the fixed-point iteration (FPI) algorithm (\cref{alg:iterative-qsp}) with the quasi-Newton (QN) method in \cite{DongMengWhaleyEtAl2021} to find QSP phase factors for $\half f_{\text{even}}(x)$ and $\half f_{\text{odd}}(x)$ with $\tau=50, 100, 150, \cdots, 1000$. The error tolerance is $\epsilon=10^{-12}$.}
\label{fig:comparison}
\end{figure}

\begin{figure}
\centering
\begin{subfigure}{0.45\textwidth}
  \includegraphics[width=60mm]{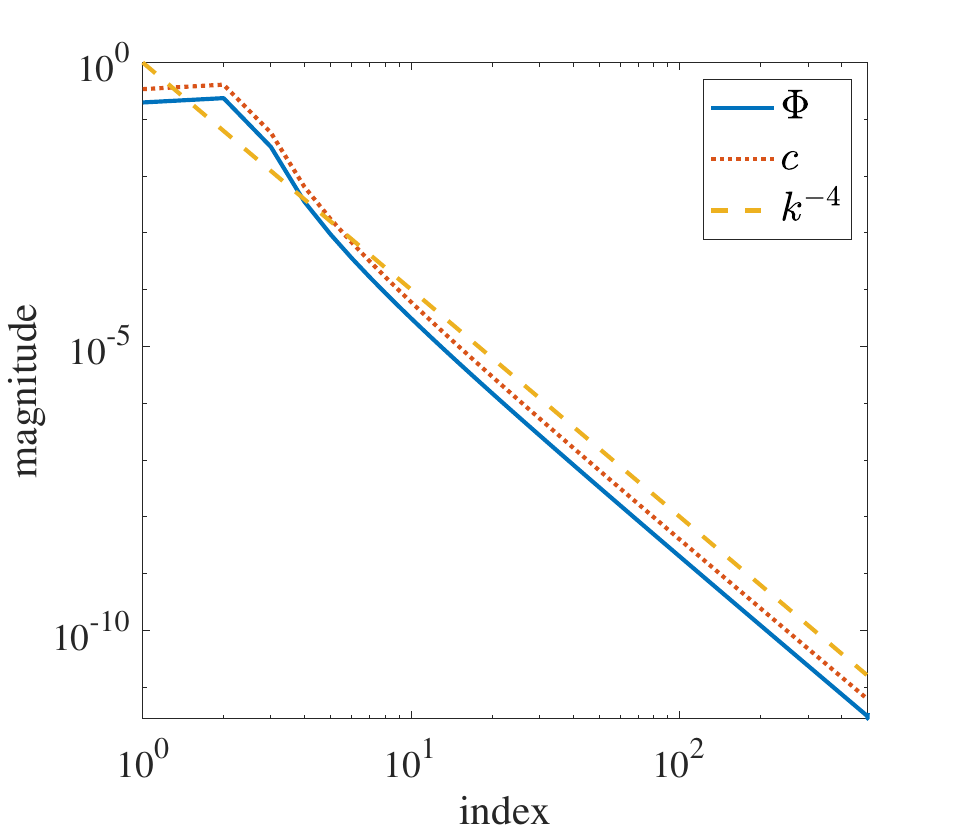}
  \caption{The target polynomial is $\sum_{k=0}^{1000}c_k T_{2k}(x)$ , where $c_k$ is the Chebyshev coefficient of $0.8\abs{x}^3$ w.r.t. $T_{2k}$. The slopes of blue and red curves are about $-4$, representing $\abs{\phi_k}\approx\text{const}\cdot k^{-4}$ and $\abs{c_k}\approx\text{const}\cdot k^{-4}$.  }
  \label{fig:decay_absx3}
\end{subfigure}
\hfill
\begin{subfigure}{0.45\textwidth}
  \includegraphics[width=60mm]{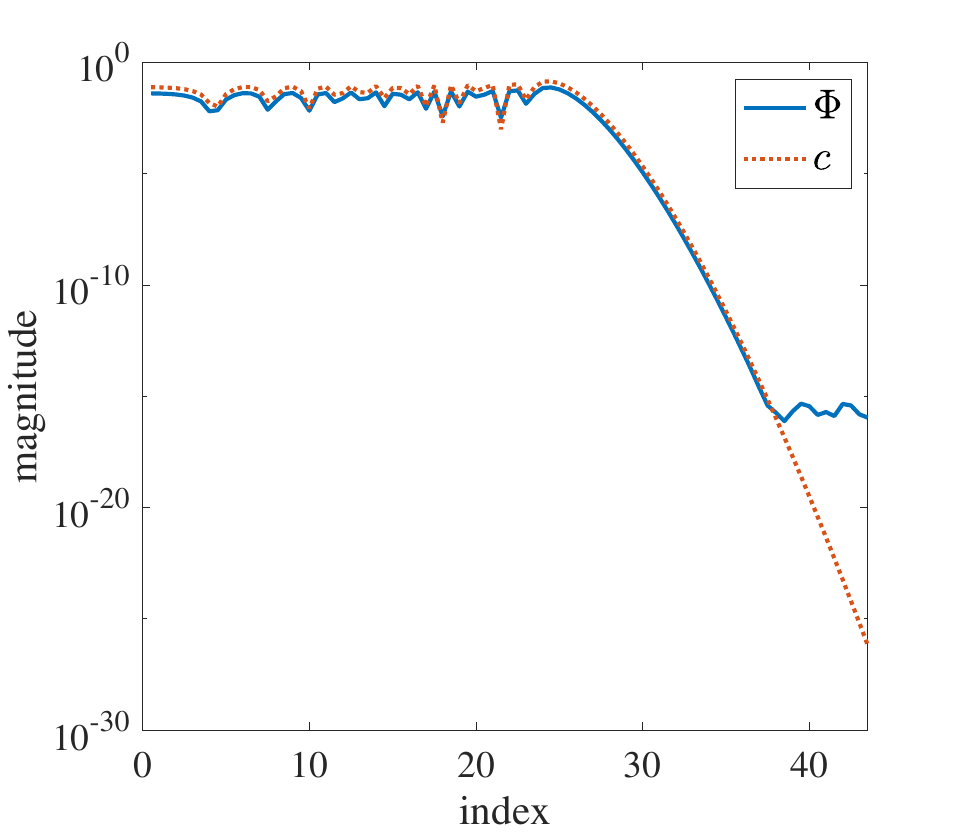}
  \caption{The target polynomial is $\half f_{\text{odd} }(x)$ with $\tau = 100$ of degree 173. }
  \label{fig:decay_HS_odd}
\end{subfigure}
\caption{Magnitude of the Chebyshev-coefficient vector $c$ and the corresponding reduced phase factors $\Phi$. }
\label{fig:numerical-pf-decay}
\end{figure}

\section{Discussion}

The question of infinite quantum processing (\cref{prob:iqsp}) asks whether there is a set of phase factors of infinite length in $\ell^1$ for representing target polynomials expressed as an infinite polynomial series. 
\cref{thm:main_l1} provides a positive answer, when the 1-norm of the Chebyshev coefficients $c$ of the target function is upper bounded by a constant $r_c$. 
While it is always possible to rescale the target function to satisfy the constraint, the constraint may be violated for many target functions without rescaling. For instance, in the Hamiltonian simulation problem, we have $\norm{c}_1=\Or(\tau)$, where $\tau$ is the simulation time and can be arbitrarily large. Numerical results in \cref{sec:numerics} indicate that both the fixed point algorithm and the decay properties persist even for large $\norm{c}_1$. 
Therefore it may be possible to significantly relax the condition $\norm{c}_1\le r_c$.

Ref.~\cite{WangDongLin2021} shows that the structure of phase factors using the infinity norm of the target polynomial $f^{(d)}$, and proves the convergence of the projected gradient method when $\norm{f^{(d)}}_{\infty}\le C d^{-1}$.
Using the tools developed in this paper, this condition can also be relaxed to $d^{-1/2}$, but the $d$-dependence may not be removed by the techniques in this paper alone.

 It is worth noting that the conditions on $\norm{f^{(d)}}_{\infty}$ and $\norm{c}_1$ are generally unrelated, i.e., neither implies the other. It is of particular interest to develop a method for computing phase factors that provably converges in the limit $\norm{f}_\infty \to 1$, which is called the fully coherent limit~\cite{MartynLiuChinEtAl2021}. 
This limit is important for the performance of certain amplitude amplification algorithms~\cite{GilyenSuLowEtAl2019,FangLinTong2022} and Hamiltonian simulation problems~\cite{MartynLiuChinEtAl2021}.

The decay properties of the reduced phase factors may also have some practical implications. For instance, for approximating smooth functions, the phase factors towards both ends of the quantum circuit are very close to being a constant. This may facilitate the compilation and error mitigation of future applications using the quantum singular value transformation.

\vspace{1em}

\noindent\textit{Note:} Since the initial posting of this paper, there have been significant algorithmic and theoretical advancements in this area. We recently proposed an algorithm based on Newton's method~\cite{DongLinNiEtAl2024}, which demonstrates rapid and robust numerical convergence for all functions admitting a QSP representation, including those where $\norm{f}_{\infty} \to 1$.   \cite{AlexisLinMnatsakanyanEtAl2024} provides the first provably stable algorithm for  finding phase factors for all polynomials with $\norm{f}_{\infty} < 1$. This is due to the recent theoretical development connecting the iQSP problem with nonlinear Fourier analysis first established in~\cite{AlexisMnatsakanyanThiele2023}.  The analysis of \cite{AlexisMnatsakanyanThiele2023} provides a positive solution to  \cref{prob:iqsp} for functions with $\norm{f}_{\infty} < 1/\sqrt{2}$, and~\cite{AlexisLinMnatsakanyanEtAl2024} generalizes this to include nearly all functions (polynomials and beyond) that admit a QSP representation.

Another recent development is the introduction of generalized quantum signal processing (GQSP)~\cite{MotlaghWeibe2023}, which relaxes the parity constraint on the target polynomial. A promising direction for future research is to investigate the potential for extending GQSP to represent smooth functions.

\vspace{1em}
\noindent\textbf{Acknowledgment}

This work was supported by the U.S. Department of Energy under the Quantum Systems Accelerator program under Grant No. DE-AC02-05CH11231 (Y.D. and L.L.), and by  the Challenge Institute for Quantum Computation (CIQC) funded by National Science Foundation (NSF) through grant number OMA-2016245 (J.W.). L.L. is a Simons Investigator in Mathematics. The authors thank discussions with Andr\'as Gily\'en and Felix Otto.

\bibliographystyle{abbrvurl}

\newpage
\appendix

\section{Some useful results related to Banach space}
\label{sec:banach}

\begin{definition}\label{dfn:C1}
Let $X,Y$ be normed vector spaces. A map $F:X\rightarrow Y$ is called $C^1$ if for every $x\in X$, there exists a (unique) bounded linear map $DF_x:X\rightarrow Y$ such that $\lim\limits_{h\rightarrow 0}\frac{\norm{F(x+h)-F(x)-DF_x(h)}}{\norm{h}}=0$, and the map $X\rightarrow \mathcal{L}(X,Y),x\mapsto DF_x$ is continuous. Here $\mathcal{L}(X,Y)$ is the set of all bounded linear maps from $X$ to $Y$ with norm topology.
\end{definition}

Here is a useful result from functional analysis.
\begin{theorem}\label{thm:invertible}
    Let $X$ be a Banach space equipped with some norm $\norm{\cdot}$ and $T:X\to X$ be a bounded linear operator. Suppose $\norm{T} < 1$, then $\bI-T$ is invertible with
    \begin{equation}
        \norm{(\bI - T)^{-1}}\le \frac{1}{1-\norm{T}},
    \end{equation}
    where $\bI$ is the identity operator.
\end{theorem}
\begin{proof}
        Consider the series $\sum_{k=0}^{\infty} \norm{T}^k = \frac{1}{1-\norm{T}}$, it converges since $\norm{T}<1$. Hence $\sum_{k=0}^{\infty} T^k$ converges in $\mathcal{L}(X,X)$ and we denote its limit as $S$. Observe that $S(\bI-T)=(\bI-T)S=\bI$. Hence $\bI-T$ is invertible and $\norm{(\bI - T)^{-1}} = \norm{S}\le \frac{1}{1-\norm{T}}$.
\end{proof}

The terminologies used in the statement of the following theorem are explained in \cref{dfn:C1}. 
\begin{theorem}\label{thm:extend}
        Suppose that $X, Y$ are dense subspaces of Banach spaces $\overline{X}, \overline{Y}$. Let $F$ be a $C^{1}$ map, $F: X \rightarrow Y$, such that for any $\left\{x_{n}\right\}$ a Cauchy sequence in $X, F\left(x_{n}\right)$ is Cauchy in $Y$ and $D F_{x_{n}}$ is Cauchy in $\mathcal{L}(X, Y)$. Then there exists a unique $C^{1}$ map $\overline{F}: \overline{X} \rightarrow \overline{Y}$ such that $\left.\overline{F}\right|_{X}=F$ and $\left.D \overline{F}\right|_{X}=D F$.
\end{theorem}
\begin{proof}
$\forall x\in\overline{X}$, there exists $\{x_n\}\subset X$ such that $x=\lim\limits_{n\to \infty}{x_n}$, then we define $\overline{F}(x)=\lim\limits_{n\to \infty} F(x_n)$. It is a well-defined and unique continuous map from $\overline{X}$ to $\overline{Y}$ extending $F$. 

To show $\overline{F}$ is $C^1$, fix $x\in X$, then $DF_x:X\to Y$ is a bounded linear map, hence it has a unique extension $G_x:\overline{X} \to \overline{Y}$. Since $F\in C^1$,$\forall \epsilon>0, \exists \delta>0$, if $h\in X$, $\norm{h}<\delta$, then 
\begin{equation}
    \norm{F(x+h)-F(x)-DF_x(h)}\leq \epsilon\norm{h}.
\end{equation}
By continuity of $\overline{F}$, we have for $\forall h\in \overline{X}$,
\begin{equation}\label{differentiable}
    \norm{\overline{F}(x+h)-\overline{F}(x)-G_x(h)}\leq \epsilon\norm{h}.
\end{equation}
Now if $x\in\overline{X}$, write $x=\lim\limits_{n\to \infty} x_n$ for $x_n \in X$, define $G_x:\overline{X}\to \overline{Y}$ by $G_x=\lim\limits_{n\to \infty} G_{x_n} $, where the limit is in norm topology. Then by assumption, $G_x$ is a well-defined, bounded linear map, and the map  $X\rightarrow \mathcal{L}(\overline{X},\overline{Y}),x\mapsto G_x$ is continuous. 

By \cref{differentiable} and the continuity of $\overline{F}(\cdot), G_\cdot(h)$, we have $\forall \epsilon>0$, $\exists \delta>0$, if $x\in \overline{X}$, $h\in \overline{X}$, $\norm{h}<\delta$,  
\begin{equation}
    \norm{\overline{F}(x+h)-\overline{F}(x)-G_x(h)}\leq \epsilon\norm{h}.
\end{equation} Hence $\overline{F}:\overline{X}\to \overline{Y}$ is $C^1$ such that $\overline{F}|_{X}=F$ and $D\overline{F}_x =G_x$, $\forall x\in \overline{X}$. 
\end{proof}

The following inverse mapping theorem can be found in e.g., {\cite[Theorem 1.2, Chapter XIV]{lang2012real}}.
\begin{theorem}[Inverse Mapping Theorem]\label{thm:inverse mapping}
        Let $X, Y$ be Banach spaces and $F: X \rightarrow Y$ be a $C^{1}$ map. Let $x_{0} \in X$ and assume that $D F_{x_{0}}$ is invertible as a bounded linear map. Then there exist open sets $U \subset X, V \subset Y$ such that $x_{0} \in U, F\left(x_{0}\right) \in V$ and $F: U \rightarrow V$ is bijective and $F^{-1}$ is $C^{1}$.
\end{theorem}

\section{Proof for \cref{lm:inverse}}\label{sec:pf_inverse}
        \begin{proof}
                The existence of $F$ is guaranteed by inverse mapping theorem. \cref{lm:DF-2I_estimate} implies that $DF(\Phi)$ is invertible for any $\Phi$ with $\norm{\Phi}_1<r_{\Phi}$. Let $\wt{F}$ be the restriction of $F$ on $B(0,r_{\Phi})$. The proof is divided into two parts:
                \begin{enumerate}
                        \item show that $\wt{F}$ is  injective. Moreover, given $\norm{\Phi}_1\leq r_\Phi$, the effective length of $\wt{F}(\Phi)$ is equal to that of $\Phi$.
                        \item show that for any $c\in \mathbb{R}^{\infty}$ with $\norm{c}_1 < r_c$, $\wt{F}^{-1}(c)$ exists.
                \end{enumerate} 
                For the first part, we prove it by contradiction. Suppose that $\wt{F}$ is not injective, then there exists $\Phi^{(1)}\ne\Phi^{(2)}$ such that $\wt{F}(\Phi^{(1)}) = \wt{F}(\Phi^{(2)})$. Denote $r = \max\left\{\norm{\Phi^{(1)}}_1,\norm{\Phi^{(2)}}_1\right\} < r_{\Phi}$, and define $\wt{G}(\Phi) := \wt{F}(\Phi) - 2\Phi$ for $\Phi \in B(0,r)$. Hence $\norm{D\wt{G}}_1=\norm{D\wt{F}-2\bI}_1\le h(r) < 2$, which derives that 
                \begin{equation}
                        \norm{\wt{G}(\Phi^{(1)}) - \wt{G}(\Phi^{(2)})}_1 \le h(r) \norm{\Phi^{(1)}-\Phi^{(2)}}_1.
                \end{equation}
                Plug in that $\wt{F}(\Phi^{(1)}) = \wt{F}(\Phi^{(2)})$, and we get $2\norm{\Phi^{(2)}-\Phi^{(1)}}_1\le h(r) \norm{\Phi^{(2)}-\Phi^{(1)}}_1$, which is a contradiction. Hence, $\wt{F}$ is injective.
                
                Next we show that the effective length of $\wt{F}(\Phi)$ is equal to that of $\Phi$ given $\norm{\Phi}_1\leq r_{\Phi}$. We use $l_{\Phi}$ and $l_{c}$ to denote the effective length of $\Phi$ and $\wt{F}(\Phi)$, respectively.  Hence, $\Phi=(\phi_0,\phi_1,\cdots, \phi_{l_{\Phi}-1},0,\cdots)$. According to \cref{thm:qsp}, $l_{\Phi}$ is no less than $l_{c}$. Suppose that $l_{\Phi}>l_{c}$. We may apply \cite[Lemma 10]{WangDongLin2021} to compute the Chebyshev coefficient of $\Im[\langle 0|U(x,\Phi(s_n))|0\rangle]$ with respect to $T_{2(l_{\Phi}-1)}$ and get 
                \begin{equation}
                    0= \sin(2\phi_{l_{\Phi}-1})\cos(2\phi_0)\prod_{i=1}^{l_{\Phi}-2}\cos^2(\phi_i).
                \end{equation}
                Hence, there exists $i\le l_{\phi}-2$ such that $\phi_i=\frac{\pi}{4}j$ for some nonzero integer $j$. Then, $\norm{\Phi}_1\geq \frac{\pi}{4}>r_\Phi$, which is a contradiction. Hence, $l_{\Phi}=l_{c}$.

                For the second part, we also prove it by contradiction. Suppose that there is an $\overline{s}<r_c$ and a $Z\in \RR^{\infty}$ with $\norm{Z}_1 = 1$ such that $\overline{s}Z$ does not lie in the range of $\wt{F}$, then we can define
                \begin{equation}
                        s^* = \inf \{s\in[0,\overline{s}]: sZ \text{\ doesn't lie in the range of } \wt{F}\}.
                \end{equation}
                From inverse mapping theorem, we know that $\wt{F}$ is invertible near $0$, hence $s^*$ is well defined. We also define $\Phi(r):=\wt{F}^{-1}(rZ)$ for any $r\in [0,s^*)$.
                
                We claim that there exists $r'<r_{\Phi}$ such that $\norm{\Phi(s)}_1<r'$ for any $s\in (0,s^*)$. That is because
                \begin{equation}
                    \begin{split}
                        s&=\norm{\wt{F}(\Phi(s))}_1=\norm{\int_{0}^1 \frac{\rd}{\rd t} \wt{F}(t\Phi(s)) \rd t}_1\\
                        &=\norm{\int_{0}^1  D\wt{F}(t\Phi(s))\cdot \Phi(s) \rd t}_1\\
                        &=\norm{\int_{0}^1  \left(2\bI+D\wt{F}(t\Phi(s))-2\bI\right)\cdot \Phi(s) \rd t}_1\\
                        &\geq \norm{\int_{0}^1  2\bI\cdot \Phi(s) \rd t}_1-\norm{\int_{0}^1  \left(D\wt{F}(t\Phi(s))-2\bI\right)\cdot \Phi(s) \rd t}_1\\
                        &\geq 2\norm{\Phi(s)}_1- \norm{\Phi(s)}_1\int_0^1\norm{D\wt{F}(t\Phi(s))-2\bI}_1 \rd t.
                    \end{split}\label{eq:F_inverse_bound}
                \end{equation}
                We apply \cref{lm:DF-2I_estimate} and get 
                \begin{equation}
                    s\geq \norm{\Phi(s)}_1\int_0^1 2- h(t\norm{\Phi(s)}_1) \rd t =H(\norm{\Phi(s)}_1).
                \end{equation}
                Note that $H(x)$ is monotonically increasing over $[0,r_{\Phi}]$ and $r_c:=H(r_{\Phi})>\overline{s}>s$. We choose $r'\in (0,r_{\Phi})$ such that $H(r')=\overline{s}$. It follows that $\norm{\Phi(s)}_1\leq r'$ for any $s\in (0,s^*)$ with $\Phi(s)\in B(0,r_{\Phi})$.
                
                Let $\{s_n\}_{n=0}^{\infty}$ be arbitrary series such that $\lim_{n\to \infty} s_n=s^*$ and $s_n<s^*$. Notice that $t\Phi(s_n)+(1-t)\Phi(s_m)\in B_1(0,r')$ for any $t\in (0,1)$. Denote $v:=\Phi(s_m)-\Phi(s_n)$ and then one has
                \begin{equation}
                \begin{split}
                   &\abs{s_n-s_m}=\norm{\wt{F}(\Phi(s_n))-\wt{F}(\Phi(s_m))}_1\\
                    &=\norm{\int_0^1 \frac{\rd}{\rd t}\wt{F}(\Phi(s_n)+tv)\rd t}_1\\
                    &=\norm{\int_0^1 D\wt{F}(\Phi(s_n)+tv)\cdot v\rd t}_1\\
                    &=\norm{\int_0^1 \left(2\bI+D\wt{F}(\Phi(s_n)+tv)-2\bI\right)\cdot v\rd t}_1\\
                    &\geq \norm{\int_0^1 2 v\rd t}_1-\norm{\int_0^1 \left(D\wt{F}(\Phi(s_n)+tv)-2\bI\right)\cdot v\rd t}_1\\
                    &\geq 2\norm{\Phi(s_n)-\Phi(s_m)}_1 -\int_0^1 \norm{\left(D\wt{F}(\Phi(s_n)+tv)-2\bI\right)\cdot v}_1\rd t\\
                    &\geq 2\norm{\Phi(s_n)-\Phi(s_m)}_1 -\int_0^1 h\left(\norm{(1-t)\Phi(s_n)+t\Phi(s_m)}_1\right) \norm{\Phi(s_n)-\Phi(s_m)}_1\rd t\\
                    &\geq 2\norm{\Phi(s_n)-\Phi(s_m)}_1 -h(r')\norm{ (\Phi(s_n)-\Phi(s_m))}_1\\
                    &\geq \left(2-h(r')\right)\norm{\Phi(s_n)-\Phi(s_m)}_1.
                \end{split}
                \end{equation}
                Since $2-h(r')>0$, we know that 
                \begin{equation}
                    \norm{\Phi(s_n)-\Phi(s_m)}_1\to 0, \quad\mathrm{as}\quad n,m\to \infty.
                \end{equation}
                According to the proof for the first part, for any $n$, the effective length of $\Phi(s_n)$ is $l_{c}$, where $l_c$ is the effective length of $Z$. We may view $\Phi(s_n)$ as vectors in $\RR^{l_{c}}$ instead. The limit of $\{\Phi(s_n)\}$ exists in $\RR^{l_{c}}$ and is unique, denoted by $\Phi^*$. Moreover, $\Phi^*$ can also be viewed as an element of $\RR^{\infty}$ and $\Phi^*\in B(0,r_{\Phi})$. By the continuity of $\wt{F}$, we know that 
                \begin{equation}
                    \wt{F}(\Phi^*)=\lim_{s\to s^*} \wt{F}(\Phi(s))=s^*Z,
                \end{equation}
                i.e., $\wt{F}^{-1}(s^*Z)$ exists.
                
                Then we can use \cref{thm:inverse mapping} at $\Phi^*$ and obtain that $\wt{F}^{-1}(sZ)$ exists for $s\in [s^*,s^*+\epsilon]$ for some $\epsilon$, which contradicts with the definition of $s^*$.

        When it comes to the proof of \cref{eq:F_inverse}, we only need to let $s = \norm{c}_1$ and $Z = \frac{c}{s}$, and do the calculation in \cref{eq:F_inverse_bound} again.
        \end{proof}

\end{document}